\documentclass[10pt,a4paper]{article}

\usepackage[T1]{fontenc}
\usepackage[utf8]{inputenc}
\usepackage{amsmath,amssymb,amsthm,amsfonts}
\usepackage[a4paper,left=3cm,right=3cm,top=4cm,bottom=3cm]{geometry}

\usepackage{graphicx}		%
\usepackage{booktabs,array}	%
\usepackage{chngcntr}		%
\counterwithin{table}{section}
\counterwithin{figure}{section}

\theoremstyle{plain}
\newtheorem{theorem}{Theorem}[section]
\newtheorem{proposition}[theorem]{Proposition}
\newtheorem{lemma}[theorem]{Lemma}
\newtheorem{corollary}[theorem]{Corollary}
\theoremstyle{definition}

\newtheorem{remark}[theorem]{Remark}
\newtheorem{example}[theorem]{Example}

\theoremstyle{remark}
\numberwithin{equation}{section}

\usepackage[pdfborder={0 0 0},bookmarksopen]{hyperref}
\hypersetup{%
	pdftitle = {Strict Local Martingales and Optimal Investment in a Black--Scholes Model with a Bubble},
	pdfauthor = {Martin Herdegen, Sebastian Herrmann},
	pdfkeywords = {},
	pdfpagemode = UseNone,
}

\newcommand{\FF}{\mathbb{F}}
\newcommand{\NN}{\mathbb{N}}
\newcommand{\RR}{\mathbb{R}}

\newcommand{\cA}{\mathcal{A}}
\newcommand{\cC}{\mathcal{C}}
\newcommand{\cE}{\mathcal{E}}
\newcommand{\cF}{\mathcal{F}}
\newcommand{\cM}{\mathcal{M}}

\newcommand{\rmm}{\mathrm{m}}
\newcommand{\rmh}{\mathrm{h}}

\newcommand{\diff}{\mathrm{d}}
\newcommand{\dd}{\,\mathrm{d}}

\newcommand{\1}{\mathbf{1}}

\newcommand*{\EX}[2][]{E^{#1}\left [ #2 \right ]}
\newcommand*{\cEX}[3][]{E^{#1}\left[ #2 \,\middle\vert\, #3 \right]}
\newcommand*{\as}[1]{#1\text{-a.s.}}
\newcommand*{\ul}[1]{\underline{#1}}

\newcommand*{\mart}[3][]{\cM_{#1}^{#2} #3}

\begin{document}
\title{%
Strict Local Martingales and Optimal Investment in a Black--Scholes Model with a Bubble%
\footnote{The authors thank J\'er\^ome Detemple, David Hobson, Johannes Muhle-Karbe, R\'emy Praz, and, in particular, Martin Schweizer for stimulating discussions and comments. We would also like to thank an associate editor and the referees for helpful comments. Moreover, we gratefully acknowledge financial support by the Swiss Finance Institute and by the National Centre of Competence in Research ``Financial Valuation and Risk Management'' (NCCR FINRISK), Project D1 (Mathematical Methods in Financial Risk Management). The NCCR FINRISK is a research instrument of the Swiss National Science Foundation.}
}
\date{}
\author{%
  Martin Herdegen%
  \thanks{%
  Department of Statistics, University of Warwick, Coventry, CV4 7AL, UK, email
  \href{mailto:M.Herdegen@warwick.ac.uk}{\nolinkurl{M.Herdegen@warwick.ac.uk}}.
  }
  \and
  Sebastian Herrmann%
  \thanks{%
  Department of Mathematics, University of Michigan, 530 Church Street, Ann Arbor, MI 48109, USA, email
  \href{mailto:sherrma@umich.edu}{\nolinkurl{sherrma@umich.edu}}.
  }
}
\maketitle

\begin{abstract}
There are two major streams of literature on the modeling of financial bubbles: the strict local martingale framework and the Johansen--Ledoit--Sornette (JLS) financial bubble model. Based on a class of models that embeds the JLS model and can exhibit strict local martingale behavior, we clarify the connection between these previously disconnected approaches. While the original JLS model is never a strict local martingale, there are relaxations which can be strict local martingales and which preserve the key assumption of a log-periodic power law for the hazard rate of the time of the crash. We then study the optimal investment problem for an investor with constant relative risk aversion in this model. We show that for positive instantaneous expected returns, investors with relative risk aversion above one always ride the bubble.
\end{abstract}

\vspace{0.5em}

{\small

\vspace{0.25em}

\vspace{0.25em}
}

\section{Introduction}
\label{sec:introduction}

Financial bubbles \cite{SornetteKaizoji2010, Protter2013, ScherbinaSchlusche2013} are often associated with a disparity between the price of an asset and its ``fundamental value''. It has been argued in the mathematical finance literature that this form of mispricing can be captured very generally by modeling asset prices as processes that are \emph{strict local martingales} (i.e., local martingales that are not martingales) under some equivalent local martingale measure (ELMM); see Loewenstein and Willard~\cite{LoewensteinWillard2000a}, Cox and Hobson~\cite{CoxHobson2005}, Heston, Loewenstein, and Willard~\cite{HestonLoewensteinWillard2007}, Jarrow, Protter, and Shimbo~\cite{JarrowProtterShimbo2007, JarrowProtterShimbo2010}, Protter~\cite{Protter2013}, and the references therein. Another strand of the literature on financial bubbles originated from the idea of fitting asset prices to a so-called log-periodic power law in order to detect and predict the end of possible bubbles; see Bouchaud, Johansen, and Sornette~\cite{BouchaudJohansenSornette1996} and Feigenbaum and Freund \cite{FeigenbaumFreund1996}. This led to the development of the Johansen--Ledoit--Sornette (JLS) financial bubble model \cite{JohansenSornette1999.criticalcrashes, JohansenLedoitSornette2000}. However, the JLS model is a martingale by definition and does not mention strict local martingales at all.

This article has two objectives: (1) to clarify the connection between these previously disconnected modeling approaches and (2) to analyze how a rational investor would act in the presence of an asset price bubble of a generalized JLS type.

\paragraph{The Johansen--Ledoit--Sornette model.}
The JLS model proposes\footnote{The following specification is taken from \cite{SornetteWoodardYanZhou2013} (up to changes in notation); the original specification in \cite{JohansenLedoitSornette2000} is slightly different and in particular has no explicit Brownian component.} that the price process of a financial asset can be modeled as the sum of its ``fundamental value'' (which is not further specified) and a \emph{bubble component} $S = (S_t)_{t\in[0,T]}$ which has the dynamics
\begin{align}
\label{eqn:intro:JLS:dynamics:original}
\frac{\diff S_t}{S_{t-}}
&= \phi'(t) \dd t + \sigma \dd W_t - \delta \dd J_t,
\end{align}
where $\phi'$ is a deterministic function, $J_t = \1_{\lbrace t \geq \gamma \rbrace}$ jumps from $0$ to $1$ at the time $\gamma$ of the crash, the constant $\delta \in (0,1)$ is the relative loss of the bubble component at the time of the crash, and $T$ is the time horizon. The time of the crash $\gamma$ is a positive random variable independent\footnote{This assumption is not explicit in the JLS model but implicit as the postulated form of the hazard rate \eqref{eqn:intro:JLS:hazard rate:LPPL} does not depend on $W$.} of the Brownian motion $W$ with a distribution function $G$ that is sufficiently regular. It is \emph{assumed} that $S$ is a (true) martingale, which in turn determines $\phi'$ via $\phi'(t) = \delta \kappa^G(t)$, $t\in (0,T)$, where $\kappa^G = G'/(1-G)$ is the hazard rate of $\gamma$.

A key assumption is that the hazard rate of $\gamma$ follows a \emph{log-periodic power law} (LPPL)
\begin{align}
\label{eqn:intro:JLS:hazard rate:LPPL}
\kappa^G(t)
&= B'\vert T-t\vert^{m-1} + C'\vert T-t\vert^{m-1} \cos \left( \varpi \log(T-t) - \psi' \right),\quad t\in(0,T),
\end{align}
where $B',C',m,T,\varpi$, and $\psi'$ are suitable real parameters; we refer to \cite[Section~2.1]{SornetteWoodardYanZhou2013} for interpretations.\footnote{The parameters have to be chosen such that the hazard rate is always nonnegative; cf.~\cite{VonBothmerMeister2003}. This constraint was ignored in many of the early articles. The ``critical time'' $T>0$ is interpreted as the end of the bubble regime \cite{SornetteWoodardYanZhou2013}, and the crash can happen at any time before $T$.} The JLS model confines the parameter $m$ to the interval $(0,1)$. This condition is equivalent to having a positive probability that the bubble does \emph{not} burst strictly before $T$ and excludes strict local martingale dynamics for $S$ (Theorem~\ref{thm:JLS:strict local martingale}). However, the justification for $m>0$ given in \cite[Section~2.2]{SornetteWoodardYanZhou2013} is debatable; see the discussion in Section~\ref{sec:relaxed JLS}. This motivates the study of a generalized JLS model. 

\paragraph{Model class and main features.}
We embed the JLS model in a larger class by relaxing some of its assumptions: $G$ may be any distribution function in $C^2[0,T)$ with $G'>0$ on $[0,T)$, the relative loss $\delta$ may be a $[0,1]$-valued deterministic function of time with $\delta(T)=0$, and $S$ may have a constant instantaneous expected return $\mu \in \RR$. In particular, the probability that the bubble does \emph{not} burst before or at $T$ can be chosen to be zero or positive. Instead of assuming that $S$ is a martingale, we only require that $S$ be a \emph{local} martingale for $\mu = 0$. The main features of this model class are:

\begin{enumerate}
\item It is \emph{flexible} enough to include specifications such that $S$ becomes a strict local martingale under a large class of ELMMs. This allows us to analyze to what extent the JLS model can be embedded in the strict local martingale framework.

\item It is \emph{tractable} enough to permit a semi-explicit solution to a utility maximization problem despite the incompleteness of the model class induced by the jump. This allows us to analyze how a rational agent should behave in the presence of an asset price bubble of this type.
\end{enumerate}

\paragraph{Objective (1): The relaxed JLS model and strict local martingales.}
The JLS model is a martingale by definition. To meet our first objective, we thus consider the \emph{relaxed JLS model} which is defined as follows: we preserve the key assumption of a log-periodic power law \eqref{eqn:intro:JLS:hazard rate:LPPL} for the hazard rate of the jump time but allow the parameter $m$ to be any real number (not necessarily in $(0,1)$), allow $\delta$ to be time-dependent in $[0,1]$ (not necessarily a constant in $(0,1)$), and only require $S$ to be a \emph{local} martingale (not necessarily a martingale) under the physical measure. We then find that the relaxed JLS model is a strict local martingale if and only if $m \leq 0$ and the function $(1-\delta)\kappa^G$ is integrable on $(0,T)$ (Theorem~\ref{thm:JLS:strict local martingale} and Remark~\ref{rem:JLS:strict local martingale}). In this case, the bubble bursts almost surely before $T$ and $\limsup_{t \uparrow\uparrow T} \delta(t) = 1$, i.e., for every $\varepsilon > 0$, there is a positive probability that the bubble component loses a fraction $1-\varepsilon$ of its value when the crash occurs.

\paragraph{Objective (2): Optimal investment.}
We study the problem of maximizing expected utility from terminal wealth for a power utility investor in the model class described above, assuming that the asset's instantaneous expected return is positive.\footnote{Korn and Wilmott \cite{KornWilmott02} study a related optimal investment problem where the distribution of the jump time and the jump size are \emph{unknown} and the optimization follows a worst-case approach; see also \cite{Seifried2010, BelakChristensenMenkens2014} and the references therein for this approach.} We provide an explicit formula for the optimal strategy and the certainty equivalent of trading in the market in terms of the solution to an integral equation (or to a first-order ODE with a nonstandard terminal condition); see Theorems~\ref{thm:main result} and \ref{thm:certainty equivalent}. The optimal strategy can be decomposed into two parts (Theorem~\ref{thm:decomposition}): a \emph{myopic demand}, which optimizes the \emph{local} performance at each point in time, and a \emph{hedging demand}, which takes into account how the dynamics of the asset price change \emph{globally} over the investor's time frame. This decomposition allows us to conclude that investors with relative risk aversion above $1$ never sell the asset short. In other words, those investors ride the bubble instead of attacking it. This theoretical insight is in line with the empirical findings of \cite{BrunnermeierNagel2004} that hedge funds were heavily invested in the stocks of the dot-com bubble despite being aware of the presence of the bubble.\footnote{\cite{TeminVoth2004} draw the same empirical conclusion from data describing the trading activities of a well-informed bank riding the South Sea bubble in 1720.}

Based on numerical illustrations, we discuss the comparative statics of the optimal strategy and the certainty equivalent. Moreover, we find that the optimal strategy is not fundamentally different when the asset price process is a strict local martingale (as opposed to the situation where it is a true martingale) under a large class of ELMMs.

\paragraph{Default risk interpretation.}
Even though the underlying economic questions are completely different, from a purely mathematical perspective, the optimal investment problem could alternatively be viewed in the context of partial \emph{default risk}. This problem has recently been studied by \cite{LimQuenez2011} and \cite{JiaoPham2011}; here, $\gamma$ is interpreted as the time of default of the risky asset. In both articles, the optimal strategy is characterized in terms of a solution to a BSDE (with jumps). In fact, our setup can be seen as a special case of \cite{JiaoPham2011}. Note, however, that our method of solving the problem (convex duality) is different from theirs (dynamic programming and BSDEs) and our solution is more explicit than theirs (in cases comparable to our setup); see \cite[Section~4.3]{JiaoPham2011}. More importantly, the convex duality approach to utility maximization is naturally linked to ELMMs. It is therefore better suited than dynamic programming for studying the strict local martingale property of the asset price process.

\paragraph{Organization of the paper.}
The rest of the paper is organized as follows. Section~\ref{sec:model class} fixes the probabilistic setup and notation, describes the model class, and explains how the JLS model and its relaxation are embedded therein. Section~\ref{sec:ELMM and strict local martingales} contains the construction of a (sub-)class of ELMMs for our financial market and presents conditions under which the asset price is a strict local martingale under such an ELMM. The optimal investment problem is studied in Section~\ref{sec:optimal investment}. Appendix~\ref{sec:change of filtration} contains a technical result that allows us to switch between certain equivalent measures and filtrations. The integral equation associated with the candidate optimal strategy is analyzed in Appendix~\ref{sec:analytic results}, while the technical aspects of the verification of its optimality are deferred to Appendix~\ref{sec:verification}.

\section{Model class}
\label{sec:model class}

Fix a finite time horizon $T>0$, and let $(\Omega,\cF,P)$ be a probability space carrying a Brownian motion $W=(W_t)_{t \in[0,T]}$ and an independent random variable $\gamma$ taking values in $(0, T]$. Define the (raw) filtrations $\FF^W = (\cF^W_t)_{t \in [0,T]}$, $\FF^\gamma = (\cF^\gamma_t)_{t \in [0,T]}$, and $\FF = (\cF_t)_{t \in [0,T]}$ by $\cF^W_t = \sigma\left(W_u: 0 \leq u \leq t \right)$, $\cF^\gamma_t = \sigma\left(\1_{\lbrace\gamma \leq u \rbrace}: 0 \leq u \leq t \right)$, and $\cF_t = \sigma\left(\cF^W_t, \cF^\gamma_t\right)$. Note that $\FF^W$ and $\FF^\gamma$ are independent under $P$ and that $\gamma$ is a stopping time with respect to $\FF^\gamma$ and $\FF$. Unless otherwise stated, all probabilistic notions requiring a probability measure and/or a filtration (e.g., (local) martingale properties of processes) pertain to $P$ and/or $\FF$. 

We denote the distribution function of $\gamma$ under $P$ by $G$ and assume that $G \in C^2[0, T)$ and $G' > 0$ on $[0, T)$; note that the law of $\gamma$ (which we denote by $\diff G$) may have a point mass at $T$, in which case $\Delta G(T) > 0$. We recall that the \emph{hazard rate} of $\gamma$ (under $P$) is the function $\kappa^G:[0,T)\to(0,\infty)$ defined by 
\begin{align}
\label{eqn:kappa}
\kappa^G(t)
&= \left(-\log(1- G(t))\right)'
= \frac{G'(t)}{1 - G(t)}.
\end{align}
It describes the conditional probability of the jump occurring in the next instant given that the jump has not happened yet. The integrability of the hazard rate is related to the existence of a point mass of $\diff G$ at $T$ as follows.

\begin{proposition}
\label{prop:hazard rate}
The following are equivalent:
\begin{enumerate}
\item The hazard rate $\kappa^G$ is nonintegrable on $(0,T)$.
\item $G(T-) = 1$.
\item $\Delta G(T) = 0$.
\end{enumerate}
\end{proposition}

\begin{proof}
As $\gamma$ is $(0,T]$-valued, $G(T) = 1$, so that the equivalence ``(b) $\Leftrightarrow$ (c)'' is trivial. Next, by the definition of $\kappa^G$, $\kappa^G(t) = -\frac{\diff}{\diff t}\log(1-G(t))$ for $t \in [0,T)$. Integrating both sides over $(0,T)$ yields
\begin{align*}
\int_0^T \kappa^G(u) \dd u
&= -\log(1-G(T-)) 
\end{align*}
(with $\log 0 := -\infty$) and proves the equivalence ``(a) $\Leftrightarrow$ (b)''.
\end{proof}

\subsection{Single jump local martingales}
\label{sec:single jump local martingales}

The asset price process in our model class is driven by the Brownian motion $W$ and a local martingale of finite variation which has a single jump at time $\gamma$. These \emph{single jump local martingales} play a major role in this paper. We introduce them here and collect some of their properties; we refer to \cite{HerdegenHerrmann2016.SJP} for a detailed study of the (local) martingale properties of this type of process.

For $F \in C^1[0, T)$, define the process $\mart{G}{F} = (\mart[t]{G}{F})_{t\in[0,T]}$ by
\begin{align}
\label{eqn:MGF}
\mart[t]{G}{F}
&= F(t) \1_{\lbrace t < \gamma \rbrace} + \cA^G F(\gamma) \1_{\lbrace t \geq \gamma \rbrace},
\end{align}
where  the function $\cA^G F:[0,T]\to\RR$ is given by
\begin{align}
\label{eqn:AGF}
\cA^G F(v)
&=
\begin{cases}
F(v) - \frac{F'(v)}{\kappa^G(v)}, &v\in[0,T),\\
F(v-)\1_{\lbrace \Delta G(T) > 0 \rbrace}, &v = T, \text{ if } \lim_{t \uparrow\uparrow T} F(t) \text{ exists in }\RR, \\
0, &v = T, \text{ if } \lim_{t \uparrow\uparrow T} F(t) \text{ does not exist.}
\end{cases}
\end{align}
Note that even though the function $F$ is only defined on the half-open interval $[0, T)$, the process $\mart{G}{F}$ is defined on the \emph{closed} interval $[0, T]$. Each trajectory $\mart[\cdot]{G}{F} (\omega)$ follows the deterministic function $F$ until just before the random time $\gamma(\omega)$, has a jump at time $\gamma$ (possibly of size $0$), and stays constant at $\cA^G F(\gamma(\omega))$ from time $\gamma(\omega)$ on.
The second and third lines in the definition \eqref{eqn:AGF} of $\cA^GF$ are only relevant if $\Delta G(T) > 0$ (otherwise $\gamma < T$ $\as{P}$). In this case, if $\mart{G}{F}$ is a local martingale, then the left limit $F(T-)$ exists in $\RR$ by Proposition~\ref{prop:local martingale property}~(b)~(i) below. This has an important implication: if $\gamma = T$, then by \eqref{eqn:MGF}--\eqref{eqn:AGF}, $\mart{G}{F}$ does not jump at all.

Under mild assumptions on $F$ and $G$, $\mart{G}{F}$ is a local $(P, \FF^\gamma)$-martingale, and so by the independence of $\FF^W$ and $\FF^\gamma$ under $P$ also a local $(P, \FF)$-martingale. The following proposition combines several results from \cite{HerdegenHerrmann2016.SJP} to provide easily checkable conditions on $F$ and $G$ for $\mart{G}{F}$ to be an integrable local martingale, a true martingale, or a square-integrable martingale with respect to the filtration $\FF^\gamma$. We stress that we may apply Proposition~\ref{prop:local martingale property} not only under $P$ but also under equivalent probability measures $Q \approx P$ on $(\Omega,\cF)$ as long as we replace $G$ by the distribution function of $\gamma$ under $Q$. Note, however, that unless $Q = P$, we can in general \emph{not} conclude that any (local) $(Q,\FF^\gamma)$-martingale is also a (local) $(Q,\FF)$-martingale. This is because $\FF^W$ and $\FF^\gamma$ can be dependent under $Q \neq P$. In this case, we have to resort to the technical ``change of filtration lemma'' (Lemma~\ref{lem:filtration} in Appendix~\ref{sec:change of filtration}) which allows us to pass from $\FF^\gamma$ to $\FF$ in certain situations.

\begin{proposition}
\label{prop:local martingale property}
Let $F \in C^1[0, T)$. 
\begin{enumerate}
\item The process $\mart{G}{F}$ is an integrable local $\FF^\gamma$-martingale if and only if $\int_0^T \vert \cA^G F(u) \vert G'(u) \dd u < \infty$. The latter condition holds if $F$ and $\cA^G F$ are bounded from below on $(0, T)$. This is automatically satisfied if $\mart{G}{F}$ is nonnegative.

\item Suppose that $\mart{G}{F}$ is a local $\FF^\gamma$-martingale.
\begin{enumerate}
\item If $\Delta G(T) > 0$, then $\mart{G}{F}$ is an $\FF^\gamma$-martingale and the limit $\lim_{t \uparrow \uparrow T} F(t)$ exists in $\RR$.
\item If $\Delta G(T) = 0$ and $\mart{G}{F}$ is integrable, then $\mart{G}{F}$ is an $\FF^\gamma$-martingale if and only if $\lim_{t\uparrow\uparrow T} F(t)(1-G(t)) = 0$.
\end{enumerate}
\item If $\int_0^T \left(\frac{F'(u)}{\kappa^G(u)}\right)^2 G'(u) \dd u < \infty$, then $\mart{G}{F}$ is a square-integrable $\FF^\gamma$-martingale.
\end{enumerate}
\end{proposition}

\begin{proof}
(a): By \cite[Lemma~3.4]{HerdegenHerrmann2016.SJP}, $\mart{G}{F}$ is integrable if and only if $\int_0^T \vert\cA^G F(u)\vert G'(u) \dd u < \infty$. Moreover, in this case $\mart{G}{F}$ is automatically a local $\FF^\gamma$-martingale by \cite[Lemma~3.7]{HerdegenHerrmann2016.SJP}. The second assertion follows from the implication \hbox{``(b)~$\Rightarrow$~(a)''} of \cite[Lemma~2.5]{HerdegenHerrmann2016.SJP}.

(b): Part~(i) follows from (a) and \cite[Theorem~3.5~(b)--(c) and Lemma~2.6]{HerdegenHerrmann2016.SJP}; part~(ii) follows from \cite[Lemma~3.7~(b)]{HerdegenHerrmann2016.SJP}.

(c): As $\int_0^T G'(u) \dd u = 1 - \Delta G(T) \leq 1$, $\int_0^T \frac{F'(u)}{\kappa^G(u)} G'(u) \dd u < \infty$ by the hypothesis and Jensen's inequality. Thus, $\mart{G}{F}$ is an $H^1$-$\FF^\gamma$-martingale by \cite[Lemmas~2.6,~3.4, and Theorem~3.5]{HerdegenHerrmann2014.strict} if $\Delta G(T) > 0$, and by \cite[Lemmas~2.8 and 3.9]{HerdegenHerrmann2014.strict} if $\Delta G(T) = 0$. Moreover, as $\mart{G}{F}$ is purely discontinuous with a single jump at $\gamma$ on $\{\gamma < T\}$, its quadratic variation satisfies
\begin{align*}
\left[\mart{G}{F}\right]_T
&= \sum_{0 < v \leq T} \left(\Delta \mart{G}{F}_v\right)^2
= \left(\frac{F'(\gamma)}{\kappa^G(\gamma)}\right)^2 \1_{\{\gamma < T\}} \;\; \as{P}
\end{align*}
In particular,
\begin{align*}
E\left[\left[\mart{G}{F}\right]_T\right]
&= \int_0^T \left(\frac{F'(u)}{\kappa^G(u)}\right)^2 G'(u) \dd u
< \infty.
\end{align*}
Thus, by the Burkholder--Davis--Gundy inequality, $\mart{G}{F}$ is a square-integrable $\FF^\gamma$-martingale.
\end{proof}

\subsection{Financial market}
\label{sec:financial market}

We consider a financial market consisting of a positive riskless asset $B = (B_t)_{t\in[0,T]}$, which is taken as the numéraire and without loss of generality normalized to $1$, and a risky asset $S = (S_t)_{t\in[0,T]}$ whose dynamics (in units of the numéraire) are given by
\begin{align}
\label{eqn:S Pdynamics}
\diff S_t
&= S_{t-} (\mu \dd t + \sigma \dd W_t + \diff \mart{G}{\phi}_t),\quad S_0 = 1.
\end{align}
Here, $\mu \in \RR$, $\sigma > 0$, and $\phi \in C^1[0,T)$ satisfies
\begin{align}
\label{eqn:standing assumption}
0
&\leq \phi'
\leq \kappa^G \text{ on } [0,T).
\end{align}
Note that \eqref{eqn:standing assumption} and Proposition~\ref{prop:local martingale property}~(c) imply that $\mart{G}{\phi}$ is a square-integrable martingale. We may assume without loss of generality that $\phi(0) = 0$. To prevent possible confusion, we stress that $S$ and $\mart{G}{\phi}$ live on the \emph{closed} interval $[0, T]$ even though $\phi$ is only defined on the half-open interval $[0, T)$.

Note that the randomness in $\mart{G}{\phi}$ stems from $\gamma$, which is interpreted as the time when the bubble bursts or the crash occurs. The dynamics of the \emph{returns process} $R = (R_t)_{t \in [0, T]}$ of $S$, defined by $R_t = \mu t + \sigma W_t + \mart[t]{G}{\phi}$, can be summarized as follows. Prior to $\gamma$, $R$ consists of a drift $(\mu + \phi'(t)) \dd t$ and a random fluctuation $\sigma \dd W_t$. Further, if $\gamma < T$, then at time $\gamma$, there is a nonpositive jump $\Delta \mart{G}{\phi}_\gamma = - \delta(\gamma) \1_{\{\gamma < T\}}$ in $R$, where $\delta: [0, T) \to [0, 1]$ defined by
\begin{align}
\label{eqn:delta}
\delta(t)
&= \phi(t) - \cA^G \phi(t)
= \frac{\phi'(t)}{\kappa^G(t)}
\end{align}
describes the absolute size of the jump of $\mart{G}{\phi}$; if $\gamma = T$, then $\mart{G}{\phi}$ does not jump (with probability $1$).\footnote{In particular, the probability that the bubble does \emph{not} burst on the interval $[0,T]$ is $\Delta G(T)$, which we allow to be nonzero.} Finally, after $\gamma$, $R$ consists of a drift $\mu \dd t$ and a random fluctuation $\sigma \dd W_t$, i.e., it satisfies the same dynamics as the returns process of a standard Black--Scholes model. Put differently, compared to the returns process of a standard Black--Scholes model with parameters $\mu$ and $\sigma$, $R$ has a nonnegative extra drift $\phi'(t) \dd t$ prior to $\gamma$, and at time $\gamma$, there is a nonpositive jump of size $-\delta(\gamma) \1_{\{\gamma < T\}}$. This models---in an idealized way---a main empirical feature of a bubble, which is a strong upward trend followed by a sharp decline at bursting. For this reason, we call $\phi'$ the \emph{instantaneous pre-crash excess return}. Moreover, we call $\mu$ the \emph{instantaneous expected return}. Using $\delta$, we can reformulate \eqref{eqn:standing assumption} as
\begin{align*}
0
&\leq \delta
\leq 1 \text{ on } [0,T),
\end{align*}
which shows that the left inequality in \eqref{eqn:standing assumption} ensures that the instantaneous pre-crash excess return is nonnegative, whereas the right inequality ensures that the stock price is always nonnegative. If the right inequality is strict for all $t \in [0, T)$, the stock price is even positive.

\subsection{The relaxed Johansen--Ledoit--Sornette model}
\label{sec:relaxed JLS}

Recall the dynamics of the bubble component in the JLS model from  \eqref{eqn:intro:JLS:dynamics:original}:
\begin{align}
\label{eqn:relaxed JLS:original dynamics}
\frac{\diff S_t}{S_{t-}}
&= \phi'(t) \dd t + \sigma \dd W_t - \delta \dd J_t,
\end{align}
where $\phi'(t) = \delta \kappa^G(t)$, $t\in(0,T)$. Using our notation for single jump local martingales, we can combine the drift term and the jump term in \eqref{eqn:relaxed JLS:original dynamics} to arrive at\footnote{To be precise, we assume here that $S$ evolves like a geometric Brownian motion after the crash; the JLS model does not specify what happens after the crash.}
\begin{align}
\label{eqn:relaxed JLS:dynamics}
\frac{\diff S_t}{S_{t-}}
&= \sigma \dd W_t + \diff \mart{G}{\phi}_t,
\end{align}
where $\mart{G}\phi$ is a single jump local martingale as introduced in \eqref{eqn:MGF} and $\phi$ is the primitive of $\phi'$ with $\phi(0)=0$, i.e.,
\begin{align*}
\phi(t)
&= \int_0^t \phi'(u) \dd u
= \delta \int_0^t \kappa^G(u) \dd u,\quad t\in[0,T).
\end{align*}
We conclude that the JLS model is a special case of \eqref{eqn:S Pdynamics} with zero instantaneous expected return $\mu = 0$, a hazard rate satisfying \eqref{eqn:intro:JLS:hazard rate:LPPL}, and $\phi'$ chosen such that $\delta(t)$, the absolute size of the jump of $\mart{G}{\phi}$ if it happens at time $t\in[0,T)$, is a constant in $(0,1)$.

\paragraph{The relaxed JLS model.}
We call $S$ a \emph{relaxed JLS model} if its dynamics are of the form \eqref{eqn:relaxed JLS:dynamics} (i.e., \eqref{eqn:S Pdynamics} with $\mu = 0$),  $\phi$ satisfies \eqref{eqn:standing assumption}, and the hazard rate $\kappa^G$ of $\gamma$ follows an LPPL \eqref{eqn:intro:JLS:hazard rate:LPPL} on $(0,T)$. The common features and differences between the JLS model and its relaxation are the following:
\begin{itemize}
\item The relaxed JLS model \emph{keeps} the general structure of the JLS model: the returns process is composed of a time-dependent drift, a Brownian motion, and a single jump process.
\item The relaxed JLS model \emph{keeps} the key assumption of an LPPL for the hazard rate of the jump time.
\item The relaxed JLS model \emph{does not require} that $S$ be a martingale; it is, however, always a local martingale by construction.
\item The relaxed JLS model \emph{does not confine} the parameter $m$ to the interval $(0,1)$; instead, $m$ can be any real number.
\item The relaxed JLS model \emph{does not require} that the relative loss $\delta$ of $S$ at the time of the crash be a constant in $(0,1)$; instead, $\delta$ is in general a $[0,1]$-valued deterministic function of time and is determined by $\phi'$ and $\kappa^G$ via \eqref{eqn:delta}.
\end{itemize}

Theorem~\ref{thm:JLS:strict local martingale} below implies that $m\leq 0$ is a necessary prerequisite for the relaxed JLS model to be a strict local martingale. Let us briefly discuss the restriction $m\in(0,1)$ imposed by the original JLS model.

\paragraph{Discussion of the restriction $m\in(0,1)$.}
In \cite[Section~2.2]{SornetteWoodardYanZhou2013}, it is argued that $m$ should lie in the interval $(0,1)$. The authors state that $m < 1$ is necessary to obtain an accelerating hazard rate. While this is certainly true, Br\'ee and Joseph~\cite{BreeJoseph2013} point out that $m < 1$ should not be an \emph{a priori} restriction when fitting the LPPL \eqref{eqn:JLS:log price:LPPL} to data. A best fit with $m \geq 1$ should rather be used to reject the model.

Here, we are concerned with the restriction $m > 0$. It is argued in \cite{SornetteWoodardYanZhou2013} that $m > 0$ is necessary to ensure that the bubble component ``remains finite at all times, including $t_c$ [$=T$]'' (p.~4419). However, we claim that if $m \leq 0$, then $\gamma < T$ $P$-a.s. Indeed, if $m \leq 0$, then the hazard rate \eqref{eqn:intro:JLS:hazard rate:LPPL} is nonintegrable on $(0,T)$, and thus $G(T-) = 1$ by Proposition~\ref{prop:hazard rate}, so that $\gamma < T$ $P$-a.s. In words, the crash happens strictly before the ``critical time'' $T$ with probability~$1$. Hence, the bubble component stays finite at all times and the argument of \cite{SornetteWoodardYanZhou2013} does not justify eliminating the case $m \leq 0$ \emph{a priori}. The authors of \cite{SornetteWoodardYanZhou2013} also claim that the property of the JLS model that there is a positive probability that no crash occurs ``makes it rational for investors to remain invested, knowing that a bubble is developing and that a crash is looming [because \dots] there is a chance for investors to gain from the ramp-up of the bubble and walk away unscathed'' (p.~4419). However, even if a crash happens almost surely before time $T$, it can similarly be argued that it is rational for investors to ride the bubble, knowing that the bubble will surely burst before time $T$, as long as they reduce their position before time $T$. With this strategy, they simply bet on the event that the bubble only bursts after they have closed their position. In fact, our Theorem~\ref{thm:decomposition} shows that investors with relative risk aversion larger than $1$ follow such a strategy as long as the underlying asset has a positive instantaneous expected return.\footnote{Investors with risk aversion less than one may also invest in the bubble under some circumstances.}$^{,}$\footnote{It is well known that risk-averse agents (with a finite credit line) never invest in an asset with zero instantaneous expected return.} We emphasize that shorting the bubble is not an arbitrage opportunity in the case where the bubble bursts almost surely before time $T$ (after all, the bubble component is a local martingale). For instance, the naive strategy of holding a (constant) short position in the bubble leads to bankruptcy with positive probability because the bubble can grow arbitrarily large if it bursts sufficiently late.

\begin{remark}
Using the formulation \eqref{eqn:relaxed JLS:dynamics}, we can also rigorously show that the log-periodic power law \eqref{eqn:intro:JLS:hazard rate:LPPL} of the hazard rate carries over to another log-periodic power law for the logarithm of the conditional expectation of the bubble component at some time $t \in (0,T)$ given the event that the crash has not yet happened.\footnote{See, e.g., \cite{VonBothmerMeister2003, BreeJoseph2013, SornetteWoodardYanZhou2013} for a formal derivation.} Using the independence of $\gamma$ and $W$, the conditional expectation of $S_t$ given that $t<\gamma$ is computed as follows:
\begin{align*}
\cEX{S_t}{t < \gamma}
&= \frac{1}{1-G(t)}\EX{S_0 \cE_t(\sigma W + \mart{G}{\phi})\1_{\lbrace t < \gamma\rbrace}}
= \frac{S_0}{1-G(t)}\EX{\cE_t(\sigma W)\exp(\phi(t))\1_{\lbrace t < \gamma\rbrace}}\\
&= S_0 \exp(\phi(t)).
\end{align*}
Hence, the logarithm of the expected value of the bubble component given that the crash has not happened yet reads as
\begin{align*}
I(t)
&:= \log \cEX{S_t}{t < \gamma}
= \log S_0 + \phi(t)
= \log S_0 + \delta \int_0^t \kappa^G(u) \dd u.
\end{align*}
Substituting the LPPL form \eqref{eqn:intro:JLS:hazard rate:LPPL} of the hazard rate, using that $m \in (0,1)$, and integrating gives
\begin{align}
\label{eqn:JLS:log price:LPPL}
I(t)
&= A + B \vert T-t \vert^m + C \vert T-t \vert^m \cos\left( \varpi  \log(T-t) - \psi\right),
\end{align}
where $B = -\delta B'/m$, $C = -\delta C'/\sqrt{m^2 + \varpi^2}$, and $A$ and $\psi$ are constants depending on $A'$, $B'$, $C'$, $m$, $T$, $\varpi$, $\psi'$, and $S_0$ (cf.~equation (6) in \cite{SornetteWoodardYanZhou2013}).

Equation~\eqref{eqn:JLS:log price:LPPL} is at the root of the literature on log-periodic power laws in the context of financial bubbles. In 1996, Bouchaud, Johansen, and Sornette~\cite{BouchaudJohansenSornette1996} and Feigenbaum and Freund~\cite{FeigenbaumFreund1996} independently suggested that the log price of a financial asset prior to a large crash can be fitted by a log-periodic power law \eqref{eqn:JLS:log price:LPPL}.\footnote{We note that no distinction between the fundamental value and the bubble component was made in the early articles. Moreover, sometimes the price is fitted instead of the log price.} The main objective is then to obtain a prediction for the ``critical time'' $T$, which is interpreted as the ``most probable time for the crash'' \cite{JohansenSornette1999.largecrashes} (because the hazard rate explodes at $T$). This approach has been widely used (see \cite{SornetteWoodardYanZhou2013,FantazziniGeraskin2013} for an overview) and intensely debated in the literature (see in particular \cite{Feigenbaum2001, JohansenSornette2001, Feigenbaum2001.more} and also \cite{BreeJoseph2013}).
\end{remark}

\begin{remark}
The case of $m=0$ has already been suggested by Ausloos, Boveroux, Minguet, and Vandewalle~\cite{AusloosBoverouxMinguetVandewalle1998.crash1987, AusloosBoverouxMinguetVandewalle1998.crash1997, AusloosBoverouxMinguetVandewalle1999}. They propose to replace the LPPL \eqref{eqn:JLS:log price:LPPL} by
\begin{align*}
I(t)
&= A + B\log(T-t) + C\log(T-t) \cos\left( \varpi  \log(T-t) - \psi\right).
\end{align*}
The corresponding hazard rate
\begin{align*}
\kappa^G(t)
&= B' \vert T-t \vert^{-1} + C' \vert T-t \vert^{-1} \cos \left( \varpi \log(T-t) - \psi' \right)
\end{align*}
is nonintegrable on $(0,T)$, and hence $\gamma < T$ $P$-a.s.~by Proposition~\ref{prop:hazard rate}. To the best of our knowledge, the case $m<0$ has not been studied in the literature so far.
\end{remark}

\section{ELMMs and the strict local martingale property}
\label{sec:ELMM and strict local martingales}

We proceed to derive a (sub-)class of ELMMs for the financial market \eqref{eqn:S Pdynamics} and to provide conditions on the model parameters for $S$ being a strict local martingale under those ELMMs. As an application, we obtain necessary and sufficient conditions for the relaxed JLS model being a strict local martingale under the physical measure.

\subsection{Preliminary results for single jump processes}
\label{sec:single jump preliminaries}

We first construct probability measures $Q^\gamma \approx P$ under which certain single jump semimartingales are square-integrable martingales. Except for the statement on square-integrability, Theorem~\ref{thm:ELMM single jump martingale} is essentially an application of the more general result \cite[Theorem~4.2]{HerdegenHerrmann2014.strict} on the existence and characterization of ELMMs for single jump semimartingales. For the convenience of the reader and because many conditions need to be checked, we provide full details.

\begin{theorem}
\label{thm:ELMM single jump martingale}
Let $F, y \in C^1[0, T)$ be such that $0 \leq F' \leq \kappa^G$ and $\inf_{t \in [0, T)} y(t) > -1$. Moreover, if $\Delta G(T) > 0$, assume that
\begin{align}
\int_0^T \vert F'(u)y(u) \vert \dd u
< \infty
\quad\text{and}\quad
\int_0^T \kappa^G (u)(1 + y(u)) \dd u
< \infty.
\label{eqn:prop:ELMM}
\end{align}
Define the functions $\zeta : [0, T) \to (0, \infty)$ and $H: [0, \infty) \to [0, 1]$ by
\begin{align}
\label{eqn:prop:ELMM:zeta}
\zeta(t)
&= \exp\left(-\int_0^t \kappa^G (u) y(u) \dd u \right),\\
\label{eqn:prop:ELMM:H}
H(t)
&= 1- \exp\left(- \int_0^t \kappa^G (u) (1 + y(u)) \dd u \right) \1_{\{t < T\}}.
\end{align}
Then $\zeta$ is positive and $\mart{G}{\zeta}$ is a positive $(P, \FF^\gamma)$-martingale starting at $1$. Define the measure $Q^\gamma \approx P$ on $\cF^\gamma_T$ by $\frac{\diff Q^\gamma}{\diff P} = \mart[T]{G}{\zeta}$. Then $\gamma$ has distribution function $H$ under $Q^\gamma$, and for $t \in [0, T)$,
\begin{align}
\label{eqn:prop:ELMM:relations:AGzeta}
\cA^G \zeta(t)
&= \zeta(t) (1+y(t)), \\
\label{eqn:prop:ELMM:relations:zeta}
1 - H(t)
&= \zeta(t)(1 - G(t)),\\
\label{eqn:prop:ELMM:relations:H}
\kappa^H(t)
:=  \frac{H'(t)}{1 - H(t)}
&= \kappa^G(t) (1+y(t)).
\end{align}
Moreover,
\begin{align}
\label{eqn:prop:ELMM:Qdynamics}
\mart{G}{F} + \int_0^{\cdot} \1_{\{u \leq \gamma \}} F'(u) y(u) \dd u
&= \mart{H}{\left(\int_0^\cdot F' (u)(1 + y(u)) \dd u\right)}
\end{align}
is a square-integrable $(Q^\gamma, \FF^\gamma)$-martingale.
\end{theorem}

\begin{proof}
We apply the more general ``removal of drift'' result \cite[Theorem~4.2]{HerdegenHerrmann2014.strict}. To this end, we define $A \in C^1[0,T)$ by $A(t) = \int_0^t F'(u) y(u) \dd u$ and declare that $0/0:=0$. Then $f := \frac{\diff F}{\diff G} = \frac{F'}{G'}$ and $a:=\frac{\diff A}{\diff G} = f y$ on $[0, T)$. Note that if $\Delta G(T) > 0$, then by \eqref{eqn:prop:ELMM},
\begin{align*}
\int_0^T \vert a(u) \vert G'(u) \dd u
&= \int_0^T \vert f(u) y(u) \vert G'(u) \dd u
= \int_0^T \vert F'(u) y(u) \vert \dd u
< \infty,\\
\int_0^T \vert f(u) \vert G'(u) \dd u
&= \int_0^T F'(u) \dd u \leq \int_0^T \kappa^G(u) \dd u
= -\log \Delta G(T)
< \infty,
\end{align*}
and so the assumptions in the first line of \cite[Theorem~4.2]{HerdegenHerrmann2014.strict} are satisfied. Moreover, clearly $\{f = 0\} \cap (0, T) \subset \{a = 0\}$, $\frac{a}{f} = y > - 1$ on $(0, T)$ and $\int_0^b \left\vert \frac{a(u)}{f(u)} \right\vert G'(u)\dd u < \infty$ for each $b \in (0,T)$, i.e., the conditions (4.8)--(4.10) (and trivially also (4.11)) in \cite{HerdegenHerrmann2014.strict} are fulfilled. 

We proceed to show that if $\Delta G(T) = 0$, then (4.13) and (4.22) for $h = 0$ in \cite{HerdegenHerrmann2014.strict} are satisfied. Indeed, the hypothesis $\inf_{t \in [0, T)} y(t) > -1$ together with the fact that $\int_0^T \kappa^G(u) \dd u = \infty$ gives
\begin{align*}
\int_0^T \left(\frac{a(u)}{f(u)}+ 1\right) \frac{G'(u)}{1- G(u)} \dd u
&=\int_0^T (1+y(u))\kappa^G(u) \dd u
= \infty.
\end{align*}

Next, we establish (4.14) and (4.15) in \cite{HerdegenHerrmann2014.strict} if $\Delta G(T) > 0$. Set $A(T) := \int_0^T F'(u) y(u) \dd u$, which is well defined by \eqref{eqn:prop:ELMM}. Then $\Delta A(T) = 0$ and so we have (4.14) in \cite{HerdegenHerrmann2014.strict}. Moreover, the identity $\frac{a}{f} = y$, \eqref{eqn:prop:ELMM} and the identity $\int_0^T \kappa^G(u) \dd u = -\log \Delta G(T) < \infty$ give 
\begin{align*}
\int_0^T \left\vert \frac{a(u)}{f(u)} \right\vert \frac{G'(u)}{1-G(u)} \dd u
&= \int_0^T \vert y(u) \vert \kappa^G(u) \dd u
< \infty.
\end{align*}
As $\frac{1}{1-G} \geq 1$ on $(0,T)$, the above yields $\int_0^T \left\vert \frac{a(u)}{f(u)} \right\vert G'(u) \dd u < \infty$, and we have condition (4.15) in \cite{HerdegenHerrmann2014.strict}.

Now, if we define $\zeta$ by (4.16) in  \cite{HerdegenHerrmann2014.strict} for $h = 0$, this simplifies to \eqref{eqn:prop:ELMM:zeta}, and the assertion about $\mart{G}{\zeta}$ follows from \cite[Theorem~4.2]{HerdegenHerrmann2014.strict}. If we define $H$ by \eqref{eqn:prop:ELMM:H}, then \eqref{eqn:prop:ELMM:relations:zeta} follows from the identity $1 - G(t) = \exp\left(-\int_0^t \kappa^G(t) \dd t\right)$, $t \in [0, T)$. Formula (4.23) in \cite{HerdegenHerrmann2014.strict} then shows that $\gamma$ has distribution function $H$ under $Q^\gamma$.\footnote{Note that $H$ is called $G^Q$ in  \cite{HerdegenHerrmann2014.strict} and that $G^Q(T) = Q[\gamma \leq T] = P[\gamma \leq T] = 1$ as $Q \approx P$.} 

Moreover, \eqref{eqn:prop:ELMM:relations:AGzeta} and \eqref{eqn:prop:ELMM:relations:H} are straightforward, and \eqref{eqn:prop:ELMM:Qdynamics} follows from assertion (4.24) in \cite{HerdegenHerrmann2014.strict}. Finally, note that the hypothesis $0 \leq F' \leq \kappa^G$ implies via \eqref{eqn:prop:ELMM:relations:H} that $0 \leq F'(1+y) \leq \kappa^H$, and so Proposition~\ref{prop:local martingale property}~(c) (with $P$ and $G$ replaced by $Q^\gamma$ and $H$, respectively) yields that $\mart{H}{\left(\int_0^\cdot F' (u)(1 + y(u)) \dd u\right)}$ is a square-integrable $(Q^\gamma, \FF^\gamma)$-martingale.
\end{proof}

It is decisive for our purposes to understand when the \emph{stochastic exponential} of the square-integrable $Q^\gamma$-martingale \eqref{eqn:prop:ELMM:Qdynamics} is a strict local martingale under $Q^\gamma$. To this end, we first provide a formula for the stochastic exponential of a single jump local martingale.

\begin{proposition}
\label{prop:stochastic exponential}
Let $F \in C^1[0, T)$ be such that $F(0) = 0$ and $0 \leq F' \leq \kappa^G$ {\normalfont($0 \leq F' < \kappa^G$)}. Then $\int_0^T \vert\cA^G (\exp \circ F)(u)\vert G'(u) \dd u < \infty$ and
\begin{align}
\label{eqn:prop:stochastic exponential}
\cE\left(\mart{G}{F} \right)
&= \mart{G}{(\exp \circ F)}
\end{align}
is a nonnegative (positive) local $(P,\FF^\gamma)$-martingale.
\end{proposition}

\begin{proof}
First, note that the assumptions $0 \leq F' \leq \kappa^G$ {($0 \leq F' < \kappa^G$)} imply that $\Delta \mart{G}{F} \geq -1$ ($\Delta \mart{G}{F} > -1$). Therefore, by the formula for the stochastic exponential (see \cite[Theorem~II~37]{Protter2005}), $\cE\left(\mart{G}{F} \right)$ is nonnegative (positive). The identity \eqref{eqn:prop:stochastic exponential} is an easy calculation. Finally, the nonnegativity of $\mart{G}{(\exp \circ F)}$ implies that $\int_0^T \vert\cA^G (\exp \circ F)(u)\vert G'(u) \dd u < \infty$ by Proposition~\ref{prop:local martingale property}~(a), which then also shows that $\cE\left(\mart{G}{F} \right)$ is an integrable local $(P,\FF^\gamma)$-martingale.
\end{proof}

The next result provides a necessary and sufficient condition for $\cE(\mart{G}{F})$ to be a strict local $(P,\FF^\gamma)$-martingale. It also shows that this strict local martingale property persists under certain changes of measure provided that the process is transformed accordingly (so that it is driftless under the new measure). 

\begin{theorem}
\label{thm:stochastic exponential:strict local martingale}
Suppose that $\Delta G(T) = 0$, $F(0) = 0$, and $0 \leq F'(t) \leq \kappa^G(t)$. Then $\cE\left(\mart{G}{F}\right)$ is a strict local $(P,\FF^\gamma)$-martingale if and only if $\int_0^T (\kappa^G(u) - F'(u)) \dd u < \infty$. Moreover, suppose that $y \in C^1[0, T)$ satisfies
\begin{align}
\label{eqn:thm:stochastic exponential:strict local martingale}
\epsilon
&\leq 1 +y(t)
\leq C + \frac{C}{F'(t)} \1_{\{\kappa^G(t) < C F'(t)\}}, \quad t \in [0, T),
\end{align}
for some constants $\epsilon \in (0,1]$ and $C \geq 1$. Define $\zeta$, $H$, and $Q^\gamma$ as in Theorem~\ref{thm:ELMM single jump martingale}. Then $\cE\left(\mart{H}{\left(\int_0^\cdot F' (u)(1 + y(u)) \dd u\right)}\right)$ is a strict local $(Q^\gamma,\FF^\gamma)$-martingale if and only if $\cE\left(\mart{G}{F}\right)$ is a strict local $(P,\FF^\gamma)$-martingale.
\end{theorem}

\begin{proof}
First, for $t \in [0, T)$, \eqref{eqn:kappa} and \eqref{eqn:prop:ELMM:relations:H} give
\begin{align}
\label{eqn:thm:stochastic exponential:strict local martingale:pf:G}
1- G(t)
&= \exp\left(-\int_0^t \kappa^G(t) \dd t\right),\\
\label{eqn:thm:stochastic exponential:strict local martingale:pf:H}
1- H(t)
&= \exp\left(-\int_0^t \kappa^H(t) \dd t\right) =  \exp\left(-\int_0^t \kappa^G(t)(1 + y(t)) \dd t\right).
\end{align}
Now, the first claim follows from Propositions~\ref{prop:stochastic exponential} and~\ref{prop:local martingale property}~(b)~(ii) and \eqref{eqn:thm:stochastic exponential:strict local martingale:pf:G} because
\begin{align*}
(\exp \circ F)(t)(1 - G(t))
&= \exp\left(-\int_0^t (\kappa^G(u)-F'(u)) \dd u\right), \quad t\in[0,T).
\end{align*}
For the second claim, note that integrability of $(\kappa^G(t)-F'(t))(1 + y(t))$ on $(0, T)$  is equivalent to integrability of $\kappa^G(t)-F'(t)$ on $(0, T)$ since by \eqref{eqn:thm:stochastic exponential:strict local martingale},
\begin{align*}
\epsilon (\kappa^G(t)-F'(t))
&\leq (\kappa^G(t)-F'(t))(1 + y(t))
\leq C (\kappa^G(t)-F'(t)) + C(C-1).
\end{align*}
Now, the second claim follows from the first one and Propositions~\ref{prop:stochastic exponential} and \ref{prop:local martingale property}~(b)~(ii) (with $Q$ and $H$ replaced by $P$ and $G$, respectively) and \eqref{eqn:thm:stochastic exponential:strict local martingale:pf:H} because
\begin{align*}
&\exp\left(\int_0^t F' (u)(1 + y(u)) \dd u\right)(1 - H(t))\\
&\;= \exp\left(-\int_0^t (\kappa^G(u)-F'(u))(1 + y(u))\dd u\right), \quad t\in[0,T).
\qedhere
\end{align*}
\end{proof}

\subsection{Equivalent local martingale measures}
\label{sec:ELMMs}

Combining the ``removal-of-drift'' result Theorem~\ref{thm:ELMM single jump martingale} for single jump semimartingales with Girsanov's theorem for Brownian motion allows us to construct a rich subclass of ELMMs for the financial market \eqref{eqn:S Pdynamics}.

\begin{theorem}
\label{thm:ELMM}
Let $y \in C^1[0, T)$ with $\inf_{t \in [0, T)} y(t) > -1$  be such that
\begin{align}
\label{eqn:thm:ELMM}
\int_0^T \left(\phi'(u) y(u)\right)^2  \dd u < \infty
\quad \text{and} \quad
\int_0^T \1_{\{\Delta G(T) > 0\}}\kappa^G (u) (1 + y(u)) \dd u < \infty.
\end{align}
Define the functions $\zeta : [0, T) \to (0, \infty)$ and $H: [0, \infty) \to [0, 1]$ and the process $Z = (Z_t)_{t \in [0, T]}$ by
\begin{align}
\label{eqn:thm:ELMM:zeta}
\zeta(t)
&= \exp\left(-\int_0^t \kappa^G (u) y(u) \dd u \right),\\
\label{eqn:thm:ELMM:H} 
H(t)
&= 1- \exp\left(- \int_0^t \kappa^G (u) (1 + y(u)) \dd u \right) \1_{\{t < T\}},\\
\label{eqn:thm:ELMM:Z}
Z_t
&= \cE_t\left ( -\int_0^\cdot \frac{1}{\sigma} \left(\mu- \phi'(u)y(u)\1_{\lbrace u \leq \gamma, u < T\rbrace}\right) \dd W_u \right ) \mart[t]{G}{\zeta}.
\end{align}
Then $Z$ is a positive $P$-martingale starting at $1$. Define the measure $Q \approx P$ on $\cF_T$ by $\frac{\diff Q}{\diff P} = Z_T$. Then $S$ is a local $Q$-martingale and satisfies the SDE
\begin{align}
\label{eqn:thm:ELMM:S Qdynamics}
\diff S_t
&= S_{t-} \left(\sigma \dd W^Q_t + \dd \mart{H}{\Big(\int_0^\cdot \phi' (u)(1 + y(u)) \dd u\Big)_t} \right),
\end{align}
where $W^Q = W +\int_0^\cdot \frac{1}{\sigma} \left(\mu- \phi'(u)y(u)\1_{\lbrace u \leq \gamma, u < T\rbrace}\right) \dd u$ is a $Q$-Brownian motion, $\gamma$ has distribution function $H$ under $Q$, and $\mart{H}{\left(\int_0^\cdot \phi' (u)(1 + y(u)) \dd u\right)}$ is a square-integrable $Q$-martingale.
\end{theorem}

\begin{proof}
For convenience, define the function $j: [0, T]^2 \to \RR$ by $j(t,v) =  \frac{1}{\sigma}\left (\mu- \phi'(t)y(t)\1_{\lbrace t \leq v, t < T\rbrace}\right)$ and set $Z^1 := \cE\left ( -\int_0^\cdot j(u, \gamma) \dd W_u \right )$ and $Z^2 := \mart{G}{\zeta}$. 

First, $Z = Z^1 Z^2$ is a ($P$, $\FF$)-martingale by Lemma~\ref{lem:filtration}~(a)~(i) with $Y^1 = Z^1$ and $Y^2 = Z^2$, using that $Z^2$ is a positive $(P,FF^\gamma)$-martingale by Theorem~\ref{thm:ELMM single jump martingale}. Clearly, $Z_0 = Z^1_0 = Z^2_0 = 1$, and $Z^2$ is also a $(P,\FF)$-martingale by the independence of $\FF^W$ and $\FF^\gamma$ under $P$.

Second, define $Q^1 \approx P$ on $\cF_T$ by $\frac{\diff Q^1}{\diff P} = Z^1_T$. Clearly, $Q^1 \approx Q$ with $\frac{\diff Q}{\diff Q^1} = Z^2_T$. By Girsanov's theorem (from $P$ to $Q^1$), $W - (-\int_0^\cdot j(u, \gamma) \dd u) = W^Q$ is a $Q^1$-Brownian motion, and again by Girsanov's theorem (from $Q^1$ to $Q$) and the fact that $Z^2$ is purely discontinuous, ${W^Q - \int_0^\cdot \frac{1}{Z^2_u} \dd [Z^2,W^Q]_u = W^Q}$ is a local $Q$-martingale. By Lévy's characterization of Brownian motion, it is even a $Q$-Brownian motion. 

Third, define $Q^\gamma \approx P$ on $\cF_T$ (and on $\cF^\gamma_T$) by $\frac{\diff Q^\gamma}{\diff P} = Z^2_T$. Then $\gamma$ has distribution function $H$ under $Q^\gamma$ by Theorem~\ref{thm:ELMM single jump martingale} and also under $Q$ by Lemma~\ref{lem:filtration}~(b)~(i), applying the latter for $X^{2, Q} = \1_{\{\gamma \leq t\}}$ and $s = 0$.

Finally, $\mart{G}{\phi} + \int_0^{\cdot} \1_{\{u \leq \gamma\}}\phi'(u) y(u) \dd u = \mart{H}{\left(\int_0^\cdot \phi' (u)(1 + y(u)) \dd u\right)}$ is a square-integrable $(Q^\gamma, \FF^\gamma)$-martingale by Theorem~\ref{thm:ELMM single jump martingale}, and hence also a square-integrable $(Q, \FF)$-martingale by Lemma~\ref{lem:filtration}~(b)~(ii). Now, \eqref{eqn:thm:ELMM:S Qdynamics} follows from the definition of $W^Q$ and the dynamics of $S$ in \eqref{eqn:S Pdynamics}.
\end{proof}

Note that under $Q$ as in Theorem~\ref{thm:ELMM}, the stock price $S$ can be written as the product of a continuous stochastic exponential and a purely discontinuous single jump local martingale. The following technical corollary to Theorem~\ref{thm:ELMM} provides conditions under which the $Q$-martingale property of the single jump local martingale carries over to the $Q$-martingale property of the product.

\begin{corollary}
\label{cor:ELMM:structure}
Let $y$, $H$, $Q$, and $W^Q$ be as in Theorem~\ref{thm:ELMM}. Let $k: [0, T]^2 \to \RR$ be of the form $k(t,v) = k_1(t) + k_2(t) \1_{\{t \leq v, t < T\}}$, where $k_1, k_2 \in L^2[0, T]$, and let $\eta \in C^1[0, T)$ be such that $\int_0^T \vert \cA^H \eta(u) \vert H'(u) \dd u < \infty$. Then 
\begin{align}
\label{eqn:cor:ELMM:structure}
\cE \left(\int_0^\cdot k(u, \gamma) \dd W_u^Q \right) \mart{H}{\eta}
\end{align}
is a local $Q$-martingale. It is a $Q$-martingale if and only if $\mart{H}{\eta}$ is a $Q$-martingale.
\end{corollary}

\begin{proof}
Let $j$, $Z$, $Z^1$, $Z^2$, and $Q^\gamma$ be as in the proof of Theorem~\ref{thm:ELMM}. Set $\tilde Z^1 := \cE \left(\int_0^\cdot k(u, \gamma) \dd W_u^Q \right)$, $\tilde Z^2 := \mart{H}{\eta}$, $Y^1 := Z^1 \tilde Z^1$, and $Y^2 := Z^2 \mart{H}{\eta}$. Then $\tilde Z^2 = \mart{H}{\eta}$ is a local $(Q^\gamma, \FF^\gamma)$-martingale by  Proposition~\ref{prop:local martingale property}~(a) (using that $\gamma$ has distribution function $H$ under $Q^\gamma$), and a short calculation gives $Y^1 = \cE \left(\int_0^\cdot (k-j)(u, \gamma) \dd W_u \right)$ $\as{P}$

We have to show that $\tilde Z^1 \tilde Z^2$ is a local $(Q, \FF)$-martingale,  and that $\tilde Z^1 \tilde Z^2$ is a $(Q, \FF)$-martingale if and only if $\tilde Z^2$ is a $(Q, \FF)$-martingale or, equivalently by Lemma~\ref{lem:filtration}~(b)~(ii), a $(Q^\gamma, \FF^\gamma)$-martingale. By Bayes' theorem, it suffices to show that $Y^1 Y^2 = Z \tilde Z^1 \tilde Z^2$ is a local $(P, \FF)$-martingale and that it is a $(P, \FF)$-martingale if and only if $Y^2 = Z^2 \tilde Z^2$ is a $(P, \FF^\gamma)$-martingale. Recalling that $Z^2_T = \frac{\diff Q^\gamma}{\diff P}$, Bayes' theorem yields that $Y^2 = Z^2 \tilde Z^2$ is a local $(P, \FF^\gamma)$-martingale, and Lemma~\ref{lem:filtration}~(a)~(ii) and (i) with $k$ replaced by $k-j$ completes the proof.
\end{proof}

\subsection{Strict local martingale conditions}
\label{sec:strict local martingale conditions}

We are now in a position to state our first main result. It gives a necessary and sufficient condition for the asset price $S$ to be a strict local martingale under certain ELMMs $Q$ constructed as in Theorem~\ref{thm:ELMM}.

\begin{theorem}
\label{thm:strict local martingale}
Let $y$ and $Q$ be as in Theorem~\ref{thm:ELMM}.
\begin{enumerate}
\item If $\Delta G(T) > 0$, then $S$ is always a $Q$-martingale.
\item If $\Delta G(T) = 0$, assume in addition that there exist constants $\epsilon \in (0, 1]$ and $C \geq 1$ such that
\begin{align}
\label{eqn:thm:strict local martingale}
\epsilon
&\leq 1 +y(t)
\leq C + \frac{C}{\phi'(t)} \1_{\{\kappa^G(t) < C \phi'(t)\}}, \quad t \in [0, T).
\end{align}
Then:
\begin{itemize}
\item $S$ is a $Q$-martingale if and only if $\int_0^T (\kappa^G(u) - \phi'(u)) \dd u = \infty$.
\item $S$ is a strict local $Q$-martingale if and only if $\int_0^T (\kappa^G(u) - \phi'(u)) \dd u < \infty$. Moreover, in this case, $\limsup_{t\uparrow\uparrow T} \delta(t) = 1$.
\end{itemize}
\end{enumerate}
\end{theorem}

\begin{proof}
First, by \eqref{eqn:thm:ELMM:S Qdynamics} and the fact that $\mart{H}{\left(\int_0^\cdot \phi' (u)(1 + y(u)) \dd u\right)}$ is purely discontinuous,
\begin{align}
\label{eqn:thm:strict local martingale:pf:S}
S
&= \cE(\sigma W^Q) \cE\left(\mart{H}{\left(\int_0^\cdot \phi' (u)(1 + y(u)) \dd u\right)}\right) \;\; \as{P}
\end{align}
Next, by Corollary~\ref{cor:ELMM:structure}, the right-hand side of  \eqref{eqn:thm:strict local martingale:pf:S} is a $Q$-martingale if and only if the second factor is. To this end, note that by Proposition~\ref{prop:stochastic exponential}, the second factor is of the form $\mart{H}{\eta}$ for $\eta = \exp \circ \left(\int_0^\cdot \phi' (u)(1 + y(u)) \dd u\right)$. Now, if $\Delta G(T) > 0$, then $\mart{H}{\eta}$ is a $Q$-martingale by Proposition~\ref{prop:local martingale property}~(b)~(i) (and Lemma~\ref{lem:filtration}~(b)~(ii) for the change of filtration). So we have (a). Otherwise, if $\Delta G(T) = 0$, Theorem~\ref{thm:stochastic exponential:strict local martingale} (and Lemma~\ref{lem:filtration}~(b)~(ii) for the change of filtration) shows that $\mart{H}{\eta}$ is a $Q$-martingale if and only if $\int_0^T (\kappa^G(u) - \phi'(u)) \dd u = \infty$. This gives both equivalences in (b).

Now, suppose that $\Delta G(T) = 0$ and that $\int_0^T (\kappa^G(u) - \phi'(u)) \dd u < \infty$. It remains to show that $\limsup_{t \uparrow\uparrow T} \delta(t) = 1$. By \eqref{eqn:standing assumption}, $\delta \leq 1$ on $(0,T)$, so it suffices to show that $\limsup_{t \uparrow\uparrow T} \delta(t) \geq 1$. Seeking a contradiction, suppose that there is $\varepsilon > 0$ such that $\limsup_{t\uparrow\uparrow T} \delta(t) \leq 1 -\varepsilon$. Then there is $t_0 \in (0,T)$ such that $\frac{\phi'(t)}{\kappa^G(t)}=\delta(t) \leq 1-\varepsilon$, $t \in (t_0,T)$, or, equivalently,
\begin{align}
\label{eqn:thm:JLS:strict local martingale:pf:30}
\kappa^G(t) - \phi'(t)
&\geq \varepsilon \kappa^G(t),\quad t\in(t_0,T).
\end{align}
Recall that $\kappa^G$ is nonintegrable on $(0,T)$ by Proposition~\ref{prop:hazard rate}. As $\kappa^G$ is continuous on $[0,T)$, it is also nonintegrable on $(t_0,T)$. But then by \eqref{eqn:thm:JLS:strict local martingale:pf:30}, also $\kappa^G - \phi'$ is nonintegrable on $(0,T)$. This is a contradiction.
\end{proof}

We illustrate Theorem~\ref{thm:strict local martingale} by giving an example where $S$ is a strict local $Q$-martingale.
\begin{example}
\label{ex:strict local martingale}
Let $\gamma$ be uniformly distributed on $[0,1]$, i.e., $T=1$, $G(t) = t$ on $[0,1]$, and let $\phi \in C^1[0,1)$ be given by $\phi(t) = -\log(1-t) - t$. Then $\phi'(t) = \frac{1}{1-t} - 1 = \kappa^G(t) - 1$ fulfills assumption \eqref{eqn:standing assumption}, and so by Theorem~\ref{thm:strict local martingale}, $S$ is a strict local martingale under any ELMM $Q$ corresponding to some $y \in C^1[0, T)$ with $\inf_{t \in [0, T)} y(t) > -1$ and satisfying \eqref{eqn:thm:ELMM} and \eqref{eqn:thm:strict local martingale} (e.g., $y\equiv 0$). Note that $\delta(t)=t$, the relative size of the jump of $S$ if it happens at time $t \in [0,1)$, increases linearly from $0$ to $1$: the later the bubble bursts, the larger the relative jump size at bursting.
\end{example}

\subsection{Strict local martingale characterization of the relaxed JLS model}
\label{sec:JLS and strict local martingales}

As an application of our first main result, Theorem~\ref{thm:strict local martingale}, we now provide necessary and sufficient conditions for the relaxed JLS model to be a strict local martingale under the physical measure. In that case, the relative jump size $\delta(t)$ must essentially converge to $1$ as $t \uparrow\uparrow T$. In other words, for each $\varepsilon > 0$, there is a positive probability that the bubble component loses a fraction $1-\varepsilon$ of its value at the time of the crash.

\begin{theorem}
\label{thm:JLS:strict local martingale}
Suppose that the hazard rate $\kappa^G$ is of the form
\begin{align}
\label{eqn:thm:JLS:strict local martingale:hazard rate}
\kappa^G(t)
&= B'\vert T-t \vert^{m-1} + C'\vert T-t \vert^{m-1} \cos \left( \varpi \log(T-t) - \psi' \right), \quad t \in [0,T),
\end{align}
for real parameters $B'$, $C'$, $m$, $\varpi$, and $\psi'$ with $\vert C'\vert < B'$ (so that $\kappa^G > 0$ on $[0,T)$). Then $S$ is a strict local martingale if and only if 
\begin{align}
\label{eqn:thm:JLS:strict local martingale:condition}
m \leq 0
\quad\text{and}\quad
\int_0^T \big(\kappa^G(u) - \phi'(u)\big)\dd u < \infty.
\end{align}
Moreover, in this case, $\limsup_{t \uparrow\uparrow T} \delta(t) = 1$.
\end{theorem}

\begin{proof}
In view of the form \eqref{eqn:thm:JLS:strict local martingale:hazard rate} for $\kappa^G$ and the property $\vert C'\vert < B'$, we first note that $m \leq 0$ is equivalent to $\kappa^G$ being nonintegrable on $(0,T)$. This together with Proposition~\ref{prop:hazard rate} shows that $m\leq 0$ if and only if $\Delta G(T) = 0$. Now, the assertions follow from Theorem~\ref{thm:strict local martingale} (with $y \equiv 0$).
\end{proof}

\begin{remark}
\label{rem:JLS:strict local martingale}
Recalling from \eqref{eqn:delta} that $\delta = \frac{\phi'}{\kappa^G}$, the second condition in \eqref{eqn:thm:JLS:strict local martingale:condition} can alternatively be formulated as $(1-\delta)\kappa^G$ being integrable on $(0,T)$.
\end{remark}

\section{Optimal investment}
\label{sec:optimal investment}

Throughout this section, we assume that $\mu > 0$ and $\phi \in C^2[0,T)$. We now analyze how a rational investor should act in the presence of an asset price bubble of the type described in Section~\ref{sec:financial market}. The optimal investment problem for a small investor is introduced in Section~\ref{sec:problem formulation}. The optimal strategy and the associated integral equation are heuristically derived in Section~\ref{sec:heuristic derivation}. Section~\ref{sec:existence and uniqueness} contains the corresponding rigorous existence and uniqueness result. A decomposition of the optimal strategy into its myopic and hedging demands as well as its economic interpretations are provided in Section~\ref{sec:myopic and hedging}. The certainty equivalent of trading in this market is computed in Section~\ref{sec:certainty equivalent}. Finally, Section~\ref{sec:numerical illustrations} presents numerical illustrations of the optimal strategy and the certainty equivalent.

\subsection{Problem formulation}
\label{sec:problem formulation}

We consider a small investor with initial capital ${x > 0}$, who can trade in the financial market described in Section~\ref{sec:financial market}. For any $\FF$-predictable, real-valued process $\pi = (\pi_t)_{t \in [0,T]}$ which is integrable with respect to the returns process  $R$, let $X^\pi =(X^\pi_t)_{t \in [0,T]}$ be the unique solution to the SDE
\begin{align}
\label{eqn:wealth process SDE}
\frac{\diff X^\pi_t}{ X^\pi_{t-}}
&= \pi_t \frac{\diff S_t}{S_{t-}}
= \pi_t \dd R_t, \quad X^\pi_0 = x.
\end{align}
We call $\pi$ an \emph{admissible strategy} if $X^\pi$ is positive. In this case, we can interpret $X^\pi$ as the \emph{wealth process} corresponding to a self-financing strategy for the market $(B,S)$ (with initial capital $x$) and $\pi_t$ as the \emph{fraction of wealth} invested in the stock at time $t$. We assume that the investor has a constant relative risk aversion with parameter $p > 0$. The corresponding utility function is given by
\begin{align*}
U(x)
&=
\begin{cases}
\frac{1}{1-p}x^{1-p} &\text{if } p \neq 1, \\
\log x & \text{if } p = 1,
\end{cases} \quad x > 0.
\end{align*}
The investor's goal is to maximize the expected utility $\EX{U(X^\pi_T)}$ over all admissible strategies $\pi$:
\begin{align}
\label{eqn:investment problem}
\EX{U(X^\pi_T)} \to \max_\pi!
\end{align}

We use the method of \emph{convex duality} both for the derivation and for the verification of the optimal strategy. Instead of the very deep general result of Kramkov and Schachermayer~\cite{KramkovSchachermayer1999} for general incomplete semimartingale models, we only use the following well-known elementary result giving a sufficient condition for optimality; cf.~the remark after \cite[Lemma~2.4]{Kallsen2000}.

\begin{proposition}
\label{prop:convex duality}
Let $\hat\pi = (\hat\pi_t)_{t\in[0,T]}$ be an admissible strategy, $\hat Q$ an equivalent local martingale measure (ELMM), and $\hat z > 0$. If
\begin{align*}
\text{\normalfont(OC1)} \quad U'(X^{\hat\pi}_T) = \hat z \frac{\diff \hat Q}{\diff P}
\qquad \text{and} \qquad
\text{\normalfont(OC2)} \quad \EX[\hat Q]{X^{\hat\pi}_T} = x,
\end{align*}
then $\hat\pi$ maximizes the expected utility $\EX{U(X^\pi_T)}$ over all admissible strategies $\pi$.
\end{proposition}
The ELMM $\hat Q$ appearing in the above result is also called the \emph{dual minimizer} corresponding to the problem \eqref{eqn:investment problem}. 

\subsection{Heuristic derivation of the optimal strategy}
\label{sec:heuristic derivation}

We proceed to derive heuristically a candidate optimal strategy $\pi$ for the investment problem \eqref{eqn:investment problem}. By virtue of Proposition~\ref{prop:convex duality}, we assume that a triplet $(\pi, Q, z)$ consisting of an admissible strategy $\pi$, an ELMM $Q$ for $S$ belonging to the class considered in Theorem~\ref{thm:ELMM}, and a number $z > 0$ satisfies the first optimality condition\footnote{Note that for the \emph{derivation} of the optimal strategy we do not need to consider the second optimality condition (OC2) in Proposition~\ref{prop:convex duality}. (OC2) is only needed for the \emph{verification}.}
\begin{align*}
\text{(OC1)} \quad U'(X^{\pi}_T)
&= z \frac{\diff Q}{\diff P}.
\end{align*}
We proceed in three steps; for ease of reading, we often drop arguments (in particular time) and do not carry out the tedious but otherwise straightforward calculations.

\emph{Step 1.} As $Q$ belongs to the class of ELMMs considered in Theorem~\ref{thm:ELMM}, there exists a nice function $y \in C^1[0,T)$ such that the density process $Z$ of $Q$ with respect to $P$ is given by
\begin{align}
\label{eqn:derivation:Z}
Z
&= \cE\left ( -\int_0^\cdot \frac{1}{\sigma} \left(\mu- \phi'y\1_{\lbrace u \leq \gamma, u < T\rbrace}\right) \dd W_u \right ) \mart{G}{\zeta},
\end{align}
where $\zeta(t) = \exp\left(-\int_0^t \kappa^G y \dd u \right)$. By (OC1), $X^{\pi}_T = (U')^{-1} (z Z_T)$, and so by  \eqref{eqn:derivation:Z}, after some algebra, 
\begin{align}
\label{eqn:derivation:Xpi:OC1}
X^{\pi}_T
&= x \cE_T \left( \int_0^\cdot \frac{1}{p \sigma} (\mu - \phi' y \1_{\lbrace u \leq \gamma, u < T \rbrace})\dd W^Q_u  \right) \times J_0  J(\gamma),
\end{align}
where $J_0 := \frac{1}{x} z^{-\frac{1}{p}} {\exp\left(\frac{1- p}{2 p^2 \sigma^2} \mu^2  T\right)}$ and the function $J:[0, T] \to (0, \infty)$ is defined by
\begin{align*}
J(v)
&=
\begin{cases}
\exp \left(\int_0^v  \frac{1 - p}{2 p^2 \sigma^2} \phi' y (\phi'y - 2\mu) + \frac{1}{p}\kappa^G y \dd u \right) (1+y(v))^{-\frac{1}{p}}  &\text{if } v < T,\\
\exp \bigg(\int_0^T  \frac{1 - p}{2 p^2 \sigma^2} \phi' y (\phi'y - 2\mu) + \frac{1}{p}\kappa^G y \dd u \bigg)  &\text{if } v = T.
\end{cases}
\end{align*}

\emph{Step 2.} By the SDE \eqref{eqn:wealth process SDE} for the wealth process $X^{\pi}$  and the dynamics \eqref{eqn:thm:ELMM:S Qdynamics} of $S$ under $Q$,
\begin{align}
\label{eqn:derivation:Xpi:MN}
X^{\pi}_T = x \cE_T\left(\int_0^{\cdot} \pi_u \sigma \dd W^Q_u \right) \times \cE_T\left(\int_0^{\cdot} \pi_u \dd \mart[u]{H}{\left(\int_0^\cdot \phi' (1 + y) \dd v\right)} \right),
\end{align}
where $H$ denotes the distribution function of $\gamma$ under $Q$. Comparing \eqref{eqn:derivation:Xpi:OC1} and \eqref{eqn:derivation:Xpi:MN}, we make the educated guess that the first and second factors as well as the integrands of the ``$\dd W^Q$-terms'' coincide. In particular, a comparison of the latter gives
\begin{align*}
\pi_t
&= \frac{1}{p \sigma^2}\left (\mu - \phi'(t)y(t)\1_{\lbrace t\leq\gamma, t< T\rbrace}\right ), \quad t \in [0, T],
\end{align*}
and it remains to determine the function $y$. As $\pi$ follows a deterministic function up to time $\gamma$, (a formal application of) Proposition~\ref{prop:stochastic exponential} gives
\begin{align*}
\cE_T\left(\int_0^{\cdot} \pi_u \dd \mart[u]{H}{\left(\int_0^\cdot \phi' (1 + y) \dd v \right)} \right)
&= \mart[T]{H}{ \xi},
\end{align*}
where $\xi(t) = \exp\big(\int_0^t \frac{1}{p \sigma^2}(\mu - \phi' y)\phi' (1 + y) \dd u \big)$, $t \in [0, T)$. Next, 
by \eqref{eqn:prop:ELMM:relations:H} and some algebra, we obtain
\begin{align}
\label{eqn:derivation:stoch exp jump}
X^{\pi}_T
&= x \cE_T\left(\int_0^{\cdot} \pi_u \sigma \dd W^Q_u \right) \times K(\gamma),
\end{align}
where the function $K:[0, T] \to (0, \infty)$ is defined by
\begin{align*}
K(v)
&=
\begin{cases}
\exp\left(\int_0^v \frac{1}{p \sigma^2}(\mu - \phi' y)\phi' (1 + y) \dd u\right) a(v,y(v),p) &\text{if } v < T,\\
\exp\Big(\int_0^T \frac{1}{p \sigma^2}(\mu - \phi' y)\phi' (1 + y) \dd u\Big) &\text{if } v = T,
\end{cases}
\end{align*}
and the function  $a: [0, T) \times [-1,\infty) \times (0, \infty) \to \RR$ is given by
\begin{align}
\label{eqn:a}
a(t,y,p)
&= 1 - \frac{1}{p \sigma^2} \frac{\phi'(t)}{\kappa^G(t)} (\mu-\phi'(t)y)
= 1- \delta(t) \frac{1}{p \sigma^2} (\mu-\phi'(t)y).
\end{align}

\emph{Step 3.} Equating the second factors on the right-hand sides of \eqref{eqn:derivation:Xpi:OC1} and \eqref{eqn:derivation:stoch exp jump} gives
\begin{align}
\label{eqn:derivation:K/J}
\frac{K(v)}{J(v)}
&= J_0, \quad v \in [0, T].
\end{align}
Using \eqref{eqn:derivation:K/J} for $v < T$ and $v = T$ and rearranging the terms yields
\begin{align}
\label{eqn:derivation:K/J:rewritten}
\frac{K(v)/K(T)}{J(v)/J(T)}
&= 1, \quad v \in [0, T).
\end{align}
Now, define the functions $b, m, n: [0, T) \times [-1,\infty) \times (0, \infty) \to \RR$ by
\begin{align}
\label{eqn:b}
b(t, y, p)
&= \Big(1 + \frac{1}{p} y\Big) a(t, y, p), \\
\label{eqn:m}
m(t,y,p)
&=(1+ y)^{\frac{1}{p}} a(t,y,p), \\
\label{eqn:n}
n(t,y,p)
&= -  \frac{1 - p}{2 p^2 \sigma^2} \left(\phi'(t) y\right)^2 + \kappa^G(t)\left(b(t, y,  p) -1\right).
\end{align}
Then after some algebra,\footnote{Note that $m$ arises from the quotient of the second factors of $K$ and $J$ and that $n$ stems from the difference of the integrands inside the exponentials of $K$ and $J$.}
\begin{align*}
\frac{K(v)/K(T)}{J(v)/J(T)}
&= m(v,y(v),p) \exp\Big(\int_v^T n(u,y(u),p) \dd u \Big), \quad v \in [0, T).
\end{align*}
Finally, plugging this into \eqref{eqn:derivation:K/J:rewritten} and rearranging the terms shows that $y$ satisfies the integral equation
\begin{align*}
m(v,y(v),p)
&= \exp\Big( - \int_v^T n(u,y(u),p) \dd u\Big),\quad v\in[0,T).
\end{align*}

\subsection{Existence and uniqueness of the optimal strategy}
\label{sec:existence and uniqueness}

We are now in a position to state our second main result. It shows that the candidate optimal strategy derived heuristically in Section~\ref{sec:heuristic derivation} exists and is indeed optimal for the utility maximization problem \eqref{eqn:investment problem}.

\begin{theorem}
\label{thm:main result}
Fix $p \in (0, \infty)$. There exists a unique function $\hat y \in C^1[0,T)$ with $\hat y > -1$ satisfying the integral equation
\begin{align}
\label{eqn:thm:main result:integral equation}
m(t,y(t),p)
&= \exp\Big( - \int_t^T n(u,y(u),p) \dd u\Big),\quad t\in[0,T).
\end{align}
The strategy $\hat\pi = (\hat\pi_t)_{t \in [0,T]}$ defined in terms of $\hat y$ by
\begin{align}
\label{eqn:thm:main result:trading strategy}
\hat\pi_t
&= \frac{1}{p \sigma^2}\left (\mu - \phi'(t) \hat y(t) \1_{\lbrace t \leq \gamma, t< T \rbrace}\right )
\end{align}
is admissible and maximizes the expected utility $\EX{U(X^\pi_T)}$ over all admissible strategies $\pi$. Moreover, $\hat y$ satisfies \eqref{eqn:thm:ELMM} and \eqref{eqn:thm:strict local martingale}.
\end{theorem}

\begin{remark}
\label{rem:main result:statement}
When we speak about a solution $\hat y$ to the integral equation \eqref{eqn:thm:main result:integral equation}, we tacitly impose that $\int_0^T \vert n(u,\hat y(u),p) \vert \dd u < \infty$. Then \eqref{eqn:thm:main result:integral equation}, the definition of $m$ in \eqref{eqn:m}, and the requirement $\hat y > -1$ imply that $a(t, \hat y(t), p) > 0$, $t \in [0,T)$. Economically, the latter property means that the investor's wealth is positive after the bubble has burst. Indeed, on $\lbrace \gamma = t \rbrace$, as the stock loses a fraction $\delta(t)$ of its value at time $t$, the wealth at time $t$ is given by
\begin{align*}
(1- \hat \pi_t)X^{\hat \pi}_{t-} + \hat \pi_t X^{\hat \pi}_{t-}(1-\delta(t))
&= X^{\hat \pi}_{t-} \left ( 1 - \delta(t) \hat \pi_t \right)
= X^{\hat\pi}_{t-}a(t, \hat y(t), p).
\end{align*}
\end{remark}

\begin{proof}[Proof of Theorem~\ref{thm:main result}]
The idea is to construct a triplet $(\hat\pi, \hat Q, \hat z)$ which satisfies the assumptions of Proposition~\ref{prop:convex duality} and thereby yields an optimizer for the investment problem \eqref{eqn:investment problem}. We proceed in three steps: first, we construct a (unique) solution to the integral equation \eqref{eqn:thm:main result:integral equation}; second, we construct a triplet $(\hat\pi, \hat Q, \hat z)$; third, we verify that this triplet satisfies the conditions (OC1) and (OC2) of Proposition~\ref{prop:convex duality}.

\emph{Step 1.} Theorem~\ref{thm:integral equation} shows in full detail that \eqref{eqn:thm:main result:integral equation} has a unique solution $\hat y > -1$ satisfying \eqref{eqn:thm:ELMM} and \eqref{eqn:thm:strict local martingale}. Here, we only outline the main difficulties and ideas. By taking logarithms on both sides, differentiating with respect to $t$ and rearranging the terms, the integral equation \eqref{eqn:thm:main result:integral equation} is easily transformed into an ODE of the form
\begin{align}
\label{eqn:thm:main result:pf:ODE}
y'(t)
&= f(t,y(t)), \quad t \in [0,T).
\end{align}
It is important to note that since \eqref{eqn:thm:main result:integral equation} need not be defined for $t=T$, also $f$ need not be defined for $t = T$. However, formally letting $t \uparrow\uparrow T$ in \eqref{eqn:thm:main result:integral equation}, we find the ``terminal condition''
\begin{align}
\label{eqn:thm:main result:pf:terminal condition}
\lim_{t \uparrow\uparrow T} m(t,y(t),p)
&= 1.
\end{align}
The fact that this ``terminal condition'' both is implicit and can only be expressed as a limit renders the ODE nonstandard. Proving existence of a solution $\hat y$ to the ODE \eqref{eqn:thm:main result:pf:ODE} can, however, be reduced to finding a pair $(y^*,y_*)$ of so-called \emph{backward upper} and \emph{backward lower solutions} to \eqref{eqn:thm:main result:pf:ODE} (cf.~Lemma~\ref{lem:general existence result}). The construction of suitable $y^*$ and $y_*$ so that the solution $\hat y$ also satisfies \eqref{eqn:thm:main result:pf:terminal condition} is the main technical difficulty of this first step of the proof.

\emph{Step 2.} Now, we construct a triplet $(\hat\pi, \hat Q, \hat z)$ as follows. First, by Step 1, $\hat y$ satisfies the assumptions of Theorem~\ref{thm:ELMM} (note that \eqref{eqn:thm:strict local martingale} implies that ${\inf_{t \in [0,T)} \hat y(t) > -1}$), which yields an explicit ELMM $\hat Q$ for $S$. Second, we have to check that $\hat\pi$ defined in \eqref{eqn:thm:main result:trading strategy} is integrable with respect to the returns process $R$ and that it is admissible. The first assertion is clear from the fact that $\hat y$ satisfies \eqref{eqn:thm:ELMM}. For the second assertion, Lemma~\ref{lem:wealth process} identifies the wealth process $X^{\hat\pi}$ in terms of $\hat y$ and shows that it remains positive; the proof is mainly computational. Third, define $\hat z > 0$ via
\begin{align}
\label{eqn:zhat} 
\hat z^{-\frac{1}{p}}
&= x m(0,\hat y(0),p) \exp\left(- (1-p)\frac{\mu^2}{2 p^2 \sigma^2} T \right);
\end{align}
note that $m(0,\hat y(0),p) > 0$ as $\hat y$ solves \eqref{eqn:thm:main result:integral equation}.

\emph{Step 3.} The verifications of (OC1) and (OC2) are carried out in Lemmas~\ref{lem:OC1} and \ref{lem:OC2}, respectively. The major difficulty of this step of the proof is to show that the candidate wealth process $X^{\hat\pi}$ is a $\hat Q$-martingale (i.e., (OC2)). The proof of (OC1) is mainly computational.
\end{proof}

\subsection{Myopic and hedging demands of the optimal strategy}
\label{sec:myopic and hedging}

A frequent goal in the context of optimal investment problems is to understand the qualitative behavior of the optimal strategy. To this end, optimal strategies are often decomposed into the sum of a myopic demand and a hedging demand (see, e.g., \cite[Section~III]{AdlerDetemple1988}, \cite[Equation~(19)]{KimOmberg1996}, \cite[Equation~(14)]{ChackoViceira2005}, \cite[Corollary~3]{Liu2007}). In discrete time, the myopic demand is the optimal strategy of an investor who treats each period as if it were the last, irrespective of the conditional distribution of any future returns (cf.~Mossin~\cite{Mossin1968}). In a continuous-time setting, the myopic demand at time $t$ can be defined as the limit (if it exists) of the optimal strategy when the investment horizon $T-t$ goes to zero. One can show that in our setting, this corresponds to letting $T \downarrow t$ in the integral equation \eqref{eqn:thm:main result:integral equation} (as one would expect formally). So the solution to the limiting equation \eqref{eqn:thm:decomposition:y^m} below can be used to define the \emph{myopic demand} via \eqref{eqn:thm:decomposition:pi^m}. Then the \emph{hedging demand} is defined as the difference between the optimal strategy and the myopic demand (cf.~\eqref{eqn:thm:decomposition:pi^h}). The following theorem states interesting consequences of this decomposition.

\begin{theorem}
\label{thm:decomposition}
Fix $p \in (0, \infty)$. There exists a unique function $y^\rmm \in C^1[0,T)$ with $y^\rmm \geq 0$ satisfying the equation
\begin{align}
\label{eqn:thm:decomposition:y^m}
m(t,y^\rmm(t),p)
&= 1.
\end{align}
Let $\hat y$ be as in Theorem~\ref{thm:main result}. The processes $\pi^\rmm = (\pi^\rmm_t)_{t \in [0,T]}$ and $\pi^\rmh = (\pi^\rmh_t)_{t \in [0,T]}$ defined in terms of $y^\rmm$ and $\hat y$ by
\begin{align}
\label{eqn:thm:decomposition:pi^m}
\pi^\rmm_t
&= \frac{1}{p \sigma^2}\left (\mu - \phi'(t)  y^\rmm(t) \1_{\lbrace t \leq \gamma, t< T \rbrace}\right ),\\
\label{eqn:thm:decomposition:pi^h}
\pi^\rmh_t
= \hat \pi_t - \pi^\rmm_t
&= \frac{1}{p \sigma^ 2} \phi'(t)(y^\rmm(t) - \hat y(t)) \1_{\lbrace t \leq \gamma, t< T \rbrace}
\end{align}
are called the \emph{myopic demand} and the \emph{hedging demand} of the optimal strategy $\hat \pi$. 
\begin{enumerate}
\item The myopic demand satisfies
\begin{align}
\label{eqn:thm:decomposition:inequality pi^m}
0
&< \pi^\rmm
\leq \frac{\mu}{p \sigma^2},
\end{align}
where on $\lbrace t \leq \gamma, t< T \rbrace$, the right inequality is an equality if and only if $\phi'(t) = 0$.
\item The hedging demand satisfies 
\begin{align}
\label{eqn:thm:decomposition:inequality pi^h}
\pi^\rmh
\leq 0 \text{ for } p \in (0, 1),
\quad
\pi^\rmh
= 0 \text{ for } p = 1,
\quad \text{and} \quad
\pi^\rmh \geq 0 \text{ for } p > 1.
\end{align}
Moreover, if $\limsup_{t \uparrow \uparrow T} G'(t) < \infty$, then $\lim_{t \uparrow \uparrow T} \pi^h_t = 0$ $\as{P}$ 
\end{enumerate}
\end{theorem}

\begin{proof}
The existence and uniqueness of $y^\rmm$ follow from Lemma~\ref{lem:implicit function}. To establish $y^\rmm \geq 0$, fix $t \in [0,T)$. By the definitions of $m$ and $a$ in \eqref{eqn:m} and \eqref{eqn:a},
\begin{align}
\label{eqn:cor:implicit function:constant 1:pf:10}
\begin{split}
1
&= m(t,y^\rmm(t), p)
= (1+y^\rmm(t))^{\frac{1}{p}} a(t,y^\rmm(t),p)\\
& = (1+y^\rmm(t))^{\frac{1}{p}} \left(1- \frac{1}{p \sigma^2}\frac{\phi'(t)}{\kappa^G(t)}(\mu - \phi'(t) y^\rmm(t)) \right).
\end{split}
\end{align}
If $y^\rmm(t) < 0$, then the right-hand side of \eqref{eqn:cor:implicit function:constant 1:pf:10} is strictly smaller than $1$, which is absurd. So $y^\rmm$ is nonnegative.

(a): Fix $t\in[0,T]$. To establish the first inequality in \eqref{eqn:thm:decomposition:inequality pi^m}, it suffices to consider the case $\phi'(t) > 0$ and $y^\rmm(t) > 0$ and $\lbrace t \leq \gamma, t< T \rbrace$; for otherwise the inequality is trivially satisfied as $\mu > 0$. In this case, by the definitions of $a$ and $m$ in \eqref{eqn:a} and \eqref{eqn:m}, the fact that $m(t, y^\rmm(t), p) = 1$, and $p > 0$, we obtain
\begin{align*}
\frac{\phi'(t)}{\kappa^G(t)} \pi^\rmm_t = \frac{1}{p \sigma^2} \frac{\phi'(t)}{\kappa^G(t)} \left(\mu - \phi'(t) y^\rmm(t)\right)
&= 1 - a(t, y^\rmm(t),p) = 1- (1+ y^\rmm(t))^{-\frac{1}{p}} > 0,
\end{align*}
and the inequality follows. The second inequality in \eqref{eqn:thm:decomposition:inequality pi^m} follows from the nonnegativity of $y^\rmm$. Finally, on $\lbrace t \leq \gamma, t< T \rbrace$, we have $\pi^\rmm_t = \frac{1}{p \sigma^2}\left (\mu - \phi'(t)y^\rmm(t) \right)$, and this is equal to $\frac{\mu}{p \sigma^2}$ if and only if $\phi'(t)y^\rmm(t) = 0$. By the nonnegativity of $y^\rmm$ and \eqref{eqn:cor:implicit function:constant 1:pf:10}, this is equivalent to $\phi'(t) = 0$.

(b): The inequalities \eqref{eqn:thm:decomposition:inequality pi^h} follow from Theorem~\ref{thm:integral equation} (noting that $y^\rmm = y_*$ for $p \leq 1$ and  $y^\rmm = y^*$ for $p \geq 1$). The second assertion is trivial if ${\Delta G(T) = 0}$. If $\Delta G(T) > 0$, Corollary~\ref{cor:implicit function} shows that $\lim_{t \uparrow \uparrow T} \phi'(t) (y^*(t) - y_*(t)) = 0$, and so a fortiori $\lim_{t \uparrow \uparrow T} \phi'(t) (y^\rmm(t) - \hat y (t)) = 0$ since $y_* \leq \hat y \leq y^*$ on $[0,T)$ by Theorem~\ref{thm:integral equation}. This completes the proof.
\end{proof}

A couple of comments are in order.

\begin{remark}
Theorem~\ref{thm:main result} shows that the optimal strategy $\hat\pi$ is generally given in terms of the solution to an integral equation (or an ODE). By contrast, to find the myopic demand of the optimal strategy, it suffices to solve an equation for each $t$.
\end{remark}

\begin{remark}
Our interpretation of the myopic demand in continuous time suggests that the hedging demand should approach $0$ at the time horizon $T$, and this holds true under a very mild technical assumption on $G$.
\end{remark}

The economic interpretation of the behavior of the hedging demand is as follows. After the bubble has burst, the model behaves like a Black--Scholes model with instantaneous expected return $\mu$ and instantaneous continuous variance $\sigma^2$. Before the crash, the instantaneous expected return is still $\mu$, but the total instantaneous variance of returns exceeds $\sigma^2$ due to the single jump component $\mart{G}{\phi}$. Hence, any risk-averse investor will favor the Black--Scholes market over our market (indeed, the certainty equivalent of trading in our market in Theorem~\ref{thm:certainty equivalent} below displays a discount with respect to the Black--Scholes certainty equivalent). The later the bubble bursts, the less time an investor has to invest in the Black--Scholes market. Consequently, it is favorable for the investor if the bubble bursts early and unfavorable if it bursts late or never.

Investors with high relative risk aversion ($p>1$) \emph{hedge against a late bursting of the bubble} with a nonnegative hedging demand $\pi^\rmh$. Indeed, in the favorable event that the bubble bursts early, they lose more money than if they had just invested myopically and profit from the (nonnegative) instantaneous pre-crash excess return $\phi'$ only for a short time. However, in the unfavorable event that the bubble bursts late or never, they profit significantly from the (nonnegative) instantaneous pre-crash excess return $\phi'$ by investing more than the myopic demand; this compensates them for the only small amount of time that remains to invest in the bubble-free market.

Investors with low relative risk aversion ($p<1$) \emph{speculate on an early bursting of the bubble} with a nonpositive hedging demand $\pi^\rmm$. Indeed, an early bursting of the bubble is favorable to them in two ways. First, as above, they can invest in the bubble-free market for a longer time period after the crash. Second, at the time of the crash, they lose less money (or even gain money in the case of a short position coming from a hedging demand that exceeds the myopic demand in absolute value) than if they had just invested myopically. However, if the bubble bursts late or never, their optimal strategy performs worse than the myopic demand, because they profit significantly less from the instantaneous pre-crash excess return $\phi'$.

In the limiting case of logarithmic utility ($p=1$), investors neither hedge against nor speculate on the timing of the crash; their optimal strategy equals the myopic demand, reflecting the well-known fact that log-investors behave myopically. Moreover, the equation $m(t,y(t),1) = 1$ reduces to a quadratic equation in $y(t)$, whose unique solution with $y > -1$ is given by $\hat y(t) = 0$ if $\phi'(t) = 0$ and
\begin{align*}
\hat y(t)
&= \frac{1}{2 \phi'(t)}\Bigg(\mu -\phi'(t) - \sigma^2\frac{\kappa(t)}{\phi'(t)} +  \sqrt{\left(\mu -\phi'(t) - \sigma^2\frac{\kappa(t)}{\phi'(t)}\right)^2 + 4 \mu \phi'(t)} \;\Bigg)
\end{align*}
if $\phi'(t) > 0$.

\subsection{Certainty equivalent}
\label{sec:certainty equivalent}

We proceed to calculate the certainty equivalent of the optimal strategy $\hat \pi$.
\begin{theorem}
\label{thm:certainty equivalent}
If $p = 1$, the certainty equivalent of trading in the market is
\begin{align}
U^{-1}\left (\EX{U(X^{\hat\pi}_T)}\right )
&= x \exp\left(\frac{\mu^2}{2 \sigma^2} T\right) \times \exp\Big(- \int_0^T \frac{\phi'(u)^2\hat y(u)^2}{2 \sigma^2}  (1 - G(u)) \dd u \Big) \notag\\
\label{eqn:thm:certainty equivalent:log utility}
&\quad\times \exp\Big(-\int_0^T \Big(\log(1 + \hat y (u)) - \frac{\hat y(u)}{1 + \hat y (u)}\Big)G'(u) \dd u\Big).
\end{align}
If $p \neq 1$, the certainty equivalent of trading in the market is
\begin{align}
\label{eqn:thm:certainty equivalent:power utility}
U^{-1}\left (\EX{U(X^{\hat\pi}_T)}\right )
&= x \exp\left (\frac{\mu^2}{2 p \sigma^2} T\right) \times m(0,\hat y(0),p)^{-\frac{p}{1-p}}.
\end{align}
\end{theorem}
The different factors in  \eqref{eqn:thm:certainty equivalent:log utility} have a clear economic interpretation. The first is the certainty equivalent of the Merton proportion $\frac{\mu}{\sigma^2}$ in the Black--Scholes model. It is shown in the proof below that the product of the first and the second factor is the certainty equivalent of the strategy $\hat \pi$ in the Black--Scholes model, so that the second factor alone describes the relative certainty equivalent loss due to trading the strategy $\hat \pi$ (instead of $\frac{\mu}{\sigma^2}$) in the Black--Scholes model. Finally, the third factor expresses the certainty equivalent loss due to the presence of the single jump component $\mart{G}{\phi}$.
 
In the case of general power utility, the first factor in \eqref{eqn:thm:certainty equivalent:power utility} is again the certainty equivalent of the Merton proportion $\frac{\mu}{p \sigma^2}$ in the Black--Scholes model, and the second one describes the combined relative certainty equivalent loss due to trading with the strategy $\hat \pi$ in the Black--Scholes model and due to the presence of the single jump component $\mart{G}{\phi}$.

\begin{proof}[Proof of Theorem~\ref{thm:certainty equivalent}]
First, assume that $p = 1$. By the definition of the wealth process and the fact that $(\mu t + \sigma W_t)_{t \in [0, 
T]}$ is a continuous semimartingale and $\mart{G}{\phi}$ a purely discontinuous martingale,
\begin{align}
\label{eqn:thm:certainty equivalent:pf:10}
X^\pi_T
&= x \cE_T\left(\int_0^\cdot \hat \pi_u \dd R_u\right)
= x \cE_T\left(\int_0^\cdot \hat \pi_u  \dd (\mu u + \sigma W_u) \right) \cE_T\left(\int_0^\cdot \hat \pi_u  \dd  \mart{G}{\phi_u} \right) \;\; \as{P}
\end{align}
We start by computing the expected value of the logarithm of the first factor on the right-hand side of \eqref{eqn:thm:certainty equivalent:pf:10}; this corresponds exactly to the utility an investor obtains from employing the strategy $\hat \pi$ in the standard  Black--Scholes model. As $\int_0^\cdot \sigma \hat \pi_u \dd W_u$ is a square-integrable martingale by the definition of $\hat \pi$ in \eqref{eqn:thm:main result:trading strategy} and \eqref{eqn:thm:integral equation:integrability properties}, a standard calculation gives
\begin{align*}
\EX{\log\left(\cE_T\left (\int_0^\cdot \hat\pi_u  \dd (\mu u + \sigma W_u) \right) \right)}
&=\EX{\frac{1}{2 \sigma^2}\int_0^T \left( \mu^2 - \phi'(u)^2 \hat y(u)^2 \1_{\{u \leq \gamma, u < T \}} \right) \dd u}\\
&= \frac{\mu^2}{2 \sigma^2} T  - \int_0^T \frac{\phi'(u)^2\hat y(u)^2}{2 \sigma^2} (1- G(u)) \dd u.
\end{align*}
To compute the expected value of the logarithm of the second factor, we first note that by the dynamics of $\mart{G}{\phi}$,
\begin{align*}
\int_0^T \hat \pi_u \dd \mart{G}{\phi}_u
&= \int_0^\gamma \bar \pi(u) \phi'(u) \dd u - \bar \pi(\gamma) \delta (\gamma) \1_{\{\gamma < T\}} \;\; \as{P},
\end{align*}
where $\bar\pi(u) := \frac{1}{p \sigma^2} \left ( \mu - \phi'(u)\hat y(u)\1_{\lbrace u < T \rbrace }  \right )$, $u\in[0,T]$.
So by the formula for the stochastic exponential,
\begin{align*}
\cE_T\left(\int_0^\cdot \hat \pi_u \dd \mart{G}{\phi}_ u\right)
&= \exp\left(\int_0^T \1_{\{u < \gamma\}} \bar \pi(u) \phi'(u) \dd u\right)\left(1 -  \bar \pi(\gamma) \delta (\gamma) \1_{\{\gamma < T\}}\right) \;\; \as{P}
\end{align*}
Thus, using the definitions of $\delta$, $a$, and $m$ in  \eqref{eqn:delta}, \eqref{eqn:a}, and \eqref{eqn:m}, the definition of $\bar\pi$, and the fact that $m(t,\hat y(t), 1) \equiv 1$ by \eqref{eqn:thm:decomposition:y^m} (since $\hat\pi = \pi^\rmm$ for $p=1$ by Theorem~\ref{thm:decomposition} (b)), for $v \in [0, T)$,
\begin{align*}
1- \bar \pi(v) \delta (v)
&= 1- \bar \pi(v) \frac{\phi' (v)}{\kappa^G(v)}
= a(v, \hat y(v),1)
= \frac{m(v, \hat y(v),1)}{1 + \hat y (v)}
= \frac{1}{1 + \hat y (v)}.
\end{align*}
The above together with the definition of $\kappa^G$ in \eqref{eqn:kappa} and the fact that $\log\left(1 -  \bar \pi(\gamma) \delta (\gamma) \1_{\{\gamma < T\}}\right) = \log \left(1 -  \bar \pi(\gamma) \delta (\gamma) \right) \1_{\{\gamma < T\}}$ gives
\begin{align*}
\EX{\log\left(\cE_T\left(\int_0^\cdot \hat \pi_u \dd \mart{G}{\phi}_ u\right)\right)}
&= \int_0^T \bar \pi(u) \phi'(u) (1 - G(u)) \dd u + \int_0^T \log\left(\frac{1}{1 + \hat y(u)}\right) G'(u) \dd u\\
&= \int_0^T \left(\bar \pi(u) \frac{\phi'(u)}{\kappa^G(u)} + \log\left(\frac{1}{1 + \hat y(u)}\right)\right) G'(u) \dd u\\
&= - \int_0^T \left( \log(1 + \hat y(u)) - \frac{\hat y(u)}{1 + \hat y(u)} \right) G'(u) \dd u.
\end{align*}
Putting everything together establishes \eqref{eqn:thm:certainty equivalent:log utility}.

Second, assume that $p \neq 1$. Then the optimality conditions (OC1) and (OC2) and the definition of $\hat z$ in \eqref{eqn:zhat} yield
\begin{align*}
\EX{U\left (X^{\hat\pi}_T\right )}
&= \frac{1}{1-p}\EX{\left (X^{\hat\pi}_T\right )^{1-p}}
= \frac{1}{1-p} \EX{X^{\hat\pi}_T U'\left(X^{\hat\pi}_T\right)}\\
&= \frac{1}{1-p} \EX{X^{\hat\pi}_T \hat z \frac{\diff \hat Q}{\diff P}}
= \hat z\frac{1}{1-p} \EX[\hat Q]{X^{\hat\pi}_T}
= \hat z \frac{1}{1-p} x\\
&= \frac{x^{1-p}}{1-p}\exp\left((1-p)\frac{\mu^2}{2 p \sigma^2} T \right )m(0,\hat y(0),p)^{-p}.\qedhere
\end{align*}
\end{proof}

\subsection{Numerical illustrations}
\label{sec:numerical illustrations}

\label{sec:illustrations}
In this section, we use numerical illustrations to answer the following four questions:
\begin{enumerate}
\item[(1)] How do shifts in the model parameters affect the optimal strategy and its myopic and hedging demands? 
\item[(2)] Can the optimal strategy involve short selling or investing more than the Merton proportion? 
\item[(3)] Does the optimal strategy distinguish fundamentally between whether or not the price process is a strict local martingale under the dual minimizer $\hat Q$?
\item[(4)] How big is the welfare loss of trading in our model in comparison to optimal investment in the Black--Scholes model? And how does the welfare loss depend on shifts in the model parameters?
\end{enumerate} 

Recall that after the bubble has burst, it is optimal to keep a constant fraction of wealth $\frac{\mu}{p \sigma^2}$ (the Merton proportion) in the stock. We thus focus on the optimal strategy \emph{before} the crash, and all plots of trading strategies show the optimal fraction of wealth invested in the stock as a function of time given that the bubble has not burst yet.

The time horizon is always $T = 1$. For questions (1) and (4), we use a cut-off exponential distribution for the jump time (in particular, with positive probability, the bubble does not burst on $[0,T]$) and a constant instantaneous pre-crash excess return $\phi'(t) = \alpha$ for different choices of $\alpha \in (0,1)$; cf.~the captions of Tables~\ref{table:illustrations:dependence on mu sigma alpha:p high}, \ref{table:illustrations:dependence on mu sigma alpha:p low}, and \ref{table:illustrations:rESRL}. To display effects corresponding to questions (2) and (3), we use other instantaneous pre-crash excess returns and/or the uniform distribution on $[0,T]$ for the jump time (under which the bubble almost surely bursts on $[0,T]$); cf.~Tables~\ref{table:illustrations:morethanmerton} and \ref{table:illustrations:strictlocalmartingale}.

\paragraph{(1) Comparative statics of the myopic and hedging demands.}

\begin{table}[!ht]
\begin{center}
\begin{tabular}{m{.3\textwidth}m{.3\textwidth}m{.3\textwidth}}
\toprule
\begin{center}
$\mu=0.05,0.1,0.2,0.3$
\end{center}&
\begin{center}
$\sigma=0.1,0.2,0.3,0.4$
\end{center}&
\begin{center}
$\alpha = 0.1, 0.2, 0.4, 0.8$
\end{center}\\
\midrule
\includegraphics[scale=0.98]{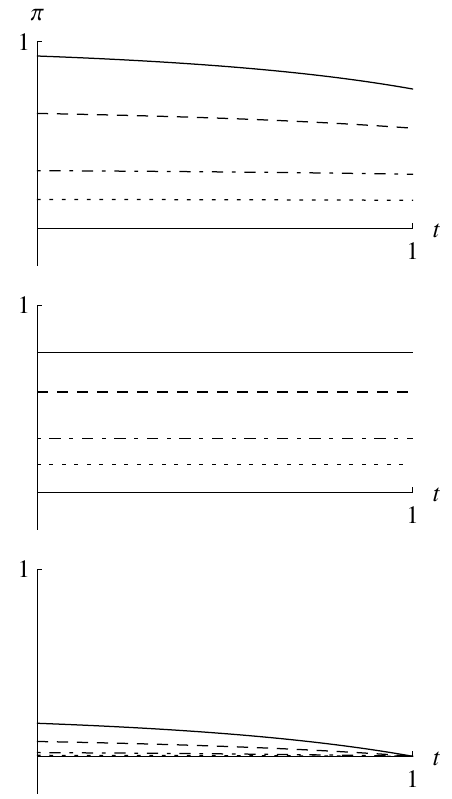}&
\includegraphics[scale=0.98]{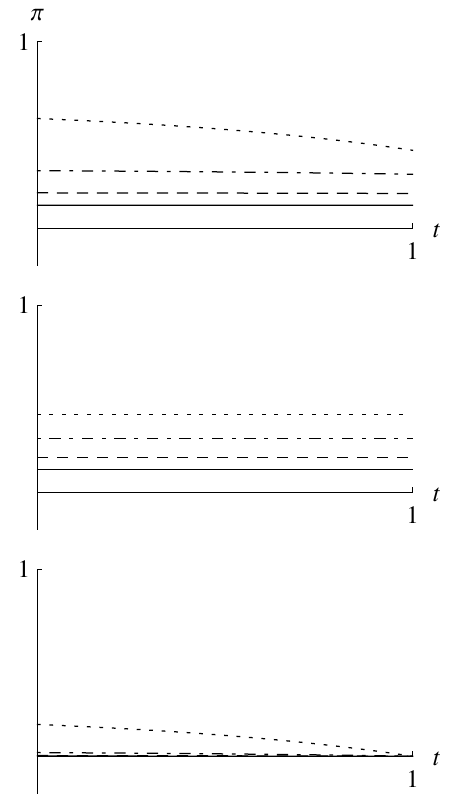}&
\includegraphics[scale=0.98]{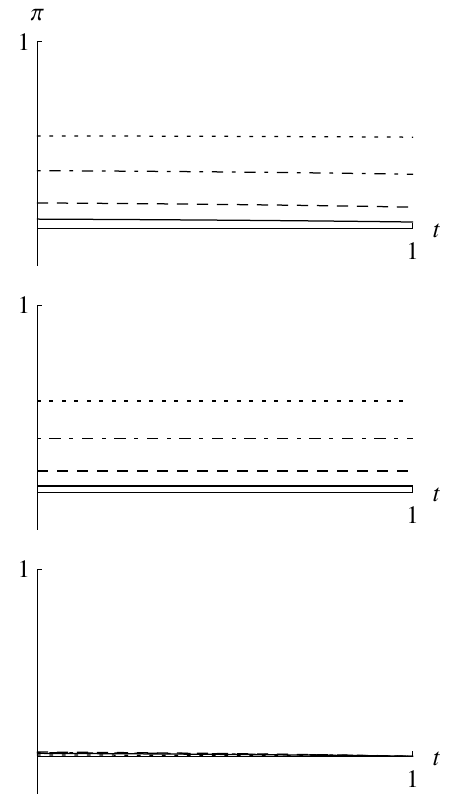}\\
\bottomrule
\end{tabular}
\caption{Optimal strategies (top row), myopic demands (middle row) and hedging demands (bottom row) for \emph{high relative risk aversion} ($p=4$); the line strength corresponds to the size of the parameter given in the head of each column (dotted lines represent the smallest value, etc.) with default parameters $\mu = 0.1$, $\sigma =0.2$, and $\alpha = 0.2$. The setup is $T=1$, $G(t)=1-\exp(-t)$, and $\phi'(t)=\alpha$; in particular, the relative jump size is $\delta(t) \equiv \alpha$.}
\label{table:illustrations:dependence on mu sigma alpha:p high}
\end{center}
\end{table}

\begin{table}[!ht]
\begin{center}
\begin{tabular}{m{.3\textwidth}m{.3\textwidth}m{.3\textwidth}}
\toprule
\begin{center}
$\mu=0.05,0.1,0.2,0.3$
\end{center}&
\begin{center}
$\sigma=0.1,0.2,0.3,0.4$
\end{center}&
\begin{center}
$\alpha = 0.1, 0.2, 0.4, 0.8$
\end{center}\\
\midrule
\includegraphics[scale=0.98]{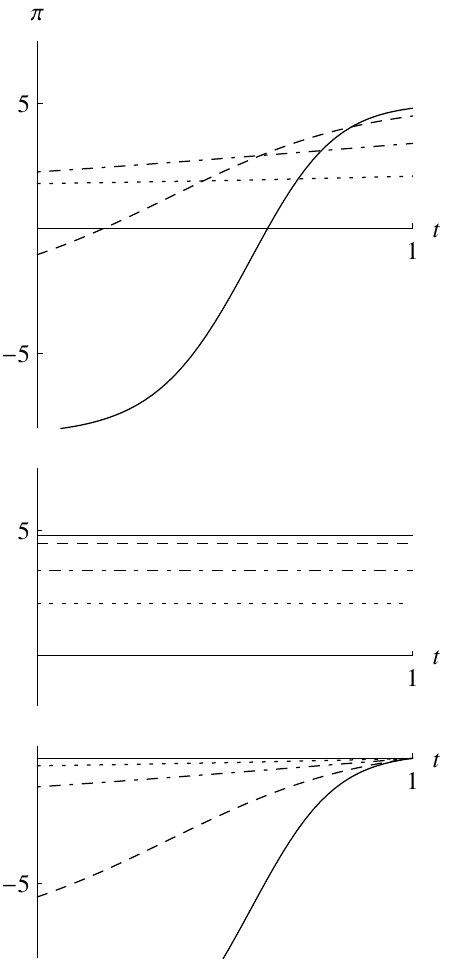}&
\includegraphics[scale=0.98]{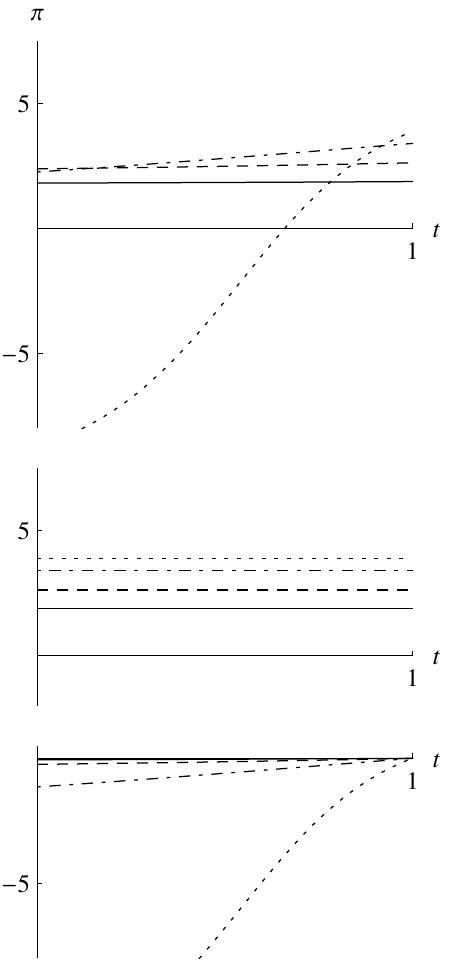}&
\includegraphics[scale=0.98]{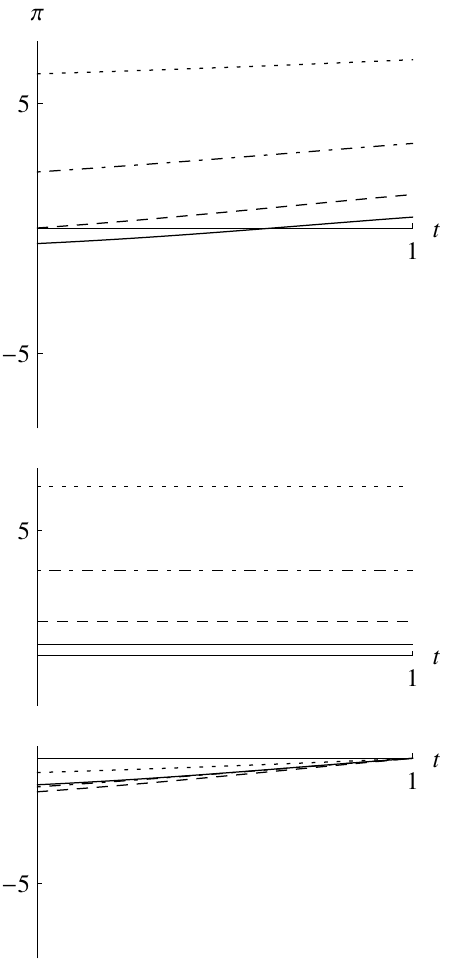}\\
\bottomrule
\end{tabular}
\caption{Same figures as in Table~\ref{table:illustrations:dependence on mu sigma alpha:p high} but for \emph{low relative risk aversion} ($p=0.25$).}
\label{table:illustrations:dependence on mu sigma alpha:p low}
\end{center}
\end{table}

Theorem~\ref{thm:decomposition} states that the sign of the hedging demand $\pi^\rmh$ is determined by the investor's relative risk aversion $p$. Thus, we provide illustrations for the cases of high ($p > 1$) and low ($p < 1$) risk aversion. The limiting case of logarithmic utility ($p=1$) always leads to a vanishing hedging demand and so the optimal strategy equals the myopic demand. It turns out that the qualitative behavior of the optimal strategy in this case closely resembles the behavior of the myopic demand of the optimal strategy in the case $p \neq 1$. We thus omit illustrations for the case $p=1$.

Tables~\ref{table:illustrations:dependence on mu sigma alpha:p high} and~\ref{table:illustrations:dependence on mu sigma alpha:p low} depict the optimal strategy before the crash as well as its decomposition into myopic and hedging demands for various choices of $\mu$, $\sigma$, and $\alpha$. Recall that $\alpha$ is a parameter describing the relative jump size of the stock price process $S$, for high ($p>1$) and low ($p<1$) risk aversion, respectively. The myopic demand is increasing in the instantaneous expected return $\mu$ and decreasing in the instantaneous continuous volatility $\sigma$ as well as in the relative jump size $\alpha$. Note that the myopic part is constant in Tables~\ref{table:illustrations:dependence on mu sigma alpha:p high} and~\ref{table:illustrations:dependence on mu sigma alpha:p low}. This is because equation \eqref{eqn:thm:decomposition:y^m} determining the myopic demand is independent of time $t$ for our choice of $G$ and $\phi'$. In general, the myopic demand need not be constant (cf.~Table~\ref{table:illustrations:morethanmerton}).

The qualitative behavior of the hedging demand, however, depends crucially on the relative risk aversion. In the case of high risk aversion ($p >1$), the hedging demand is always nonnegative and has the same monotonicity properties as the myopic demand. In the case of low risk aversion ($p<1$), the hedging demand is nonpositive and the monotonicity properties of the hedging demand are no longer in line with those of the myopic demand. Indeed, it is decreasing in $\mu$ (increasing in absolute value), increasing in $\sigma$ (decreasing in absolute value), and ``U-shaped'' in $\alpha$.

\paragraph{(2) Short selling and investing more than the Merton proportion.}

The optimal strategy includes short selling if and only if the investor's relative risk aversion $p$ is smaller than $1$ and the (nonpositive) hedging demand exceeds the (nonnegative) myopic demand in absolute value; cf.~the discussion after Theorem~\ref{thm:decomposition}. Table~\ref{table:illustrations:dependence on mu sigma alpha:p low} shows that short selling is amplified by ``good'' post-crash investment opportunities, i.e., low $\sigma$ and high $\mu$. For high relative risk aversion ($p>1$), the myopic and hedging demands are always nonnegative (by Theorem~\ref{thm:decomposition}); hence the optimal strategy never involves short selling.

\begin{table}[!ht]
\begin{center}
\begin{tabular}{m{.3\textwidth}m{.3\textwidth}m{.3\textwidth}}
\toprule
\begin{center}
optimal strategy
\end{center}&
\begin{center}
myopic demand
\end{center}&
\begin{center}
hedging demand
\end{center}\\
\midrule
\includegraphics[scale=0.98]{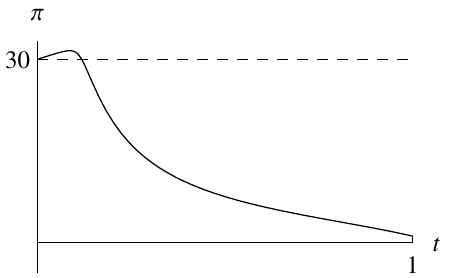}&
\includegraphics[scale=0.98]{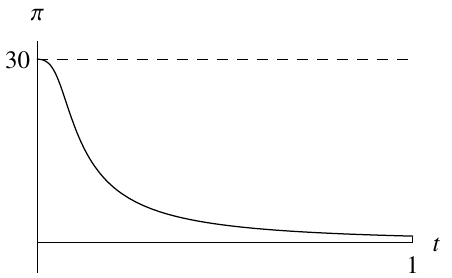}&
\includegraphics[scale=0.98]{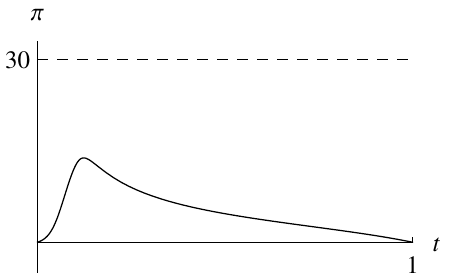}\\
\bottomrule
\end{tabular}
\caption{Under extreme circumstances, the optimal strategy before the crash (solid, left panel) may lie above the Merton proportion (dashed). The middle and right panels show the corresponding myopic and hedging demands, respectively. The setup is $T=1$, $G(t)=1-\exp(-t)$, and $\phi'(t) = 0.2 t$. The parameters are $\mu = 0.3$, $\sigma =0.05$, and $p = 4$.}
\label{table:illustrations:morethanmerton}
\end{center}
\end{table}

When $p>1$, the optimal strategy may lie above the Merton proportion (Table~\ref{table:illustrations:morethanmerton}). At first glance, this might be surprising as the instantaneous variance of our model is higher than in the corresponding Black--Scholes model due to the presence of the extra single jump component. However, on closer inspection, this effect can be explained by a combination of a high myopic demand at time $0$ and a hedging demand that is sufficiently increasing close to time $0$.

\paragraph{(3) Strict local martingales vs.~true martingales.}

\begin{table}[!ht]
\begin{center}
\begin{tabular}{m{.3\textwidth}m{.3\textwidth}m{.3\textwidth}}
\toprule
\begin{center}
optimal strategy
\end{center}&
\begin{center}
myopic demand
\end{center}&
\begin{center}
hedging demand
\end{center}\\
\midrule
\includegraphics[scale=0.98]{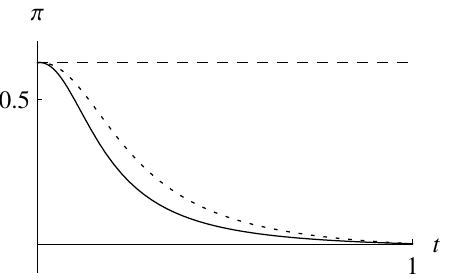}&
\includegraphics[scale=0.98]{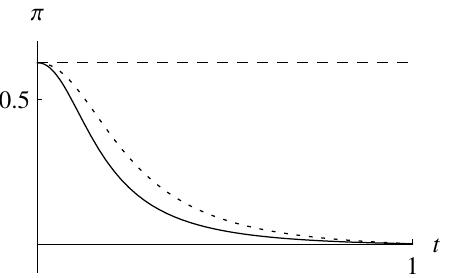}&
\includegraphics[scale=0.98]{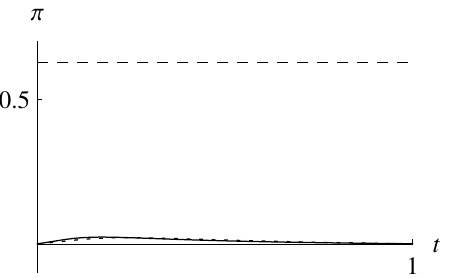}\\
\bottomrule
\end{tabular}
\caption{The optimal strategy does not distinguish qualitatively between $S$ being a strict local martingale or a true martingale under the dual minimizer $\hat Q$. The setup is $T=1$, $G(t)=t$, and $\phi'(t) = \alpha (\frac{1}{1-t} - 1)$; in particular, the relative jump size is $\delta(t)=\alpha t$. The solid lines correspond to $\alpha = 1$, for which $S$ is a strict local martingale under $\hat Q$, the dotted lines correspond to $\alpha = 0.7$, for which $S$ is a true martingale under $\hat Q$. The dashed lines represent the Merton proportion. The parameters are $\mu = 0.1$, $\sigma =0.2$, and $p = 4$.}
\label{table:illustrations:strictlocalmartingale}
\end{center}
\end{table}

The investor's optimal strategy does not seem to clearly distinguish between the asset price being a strict local martingale and a true martingale under the dual minimizer $\hat Q$. The solid lines in Table~\ref{table:illustrations:strictlocalmartingale} illustrate the optimal strategy and its decomposition into myopic and hedging demands in the case where $S$ is a strict local martingale under the dual minimizer $\hat Q$ (and in fact under any ELMM $Q$ obtained via Theorem~\ref{thm:ELMM} under the additional condition \eqref{eqn:thm:strict local martingale}). For $\alpha = 1$, the setup of Table~\ref{table:illustrations:strictlocalmartingale} coincides with Example~\ref{ex:strict local martingale}. However, for any $\alpha \in [0,1)$, the stock price process $S$ is a true martingale under $\hat Q$ (by Theorem~\ref{thm:strict local martingale}), and the dotted lines in Table~\ref{table:illustrations:strictlocalmartingale} depict the optimal strategy and its decomposition into myopic and hedging demands for $\alpha =0.7$. The graphs show that the qualitative behavior of the optimal strategies is quite similar. In fact, the optimal strategies converge (numerically) as $\alpha \uparrow 1$.

\paragraph{(4) Comparative statics of the welfare loss relative to the Black--Scholes model.}

\begin{table}[!ht]
\begin{center}
\begin{tabular}{m{.3\textwidth}m{.3\textwidth}m{.3\textwidth}}
\toprule
\begin{center}
$p=0.25$
\end{center}&
\begin{center}
$p=1$
\end{center}&
\begin{center}
$p=4$
\end{center}\\
\midrule
\includegraphics[scale=.98]{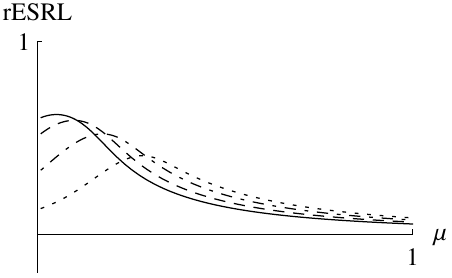}&
\includegraphics[scale=.98]{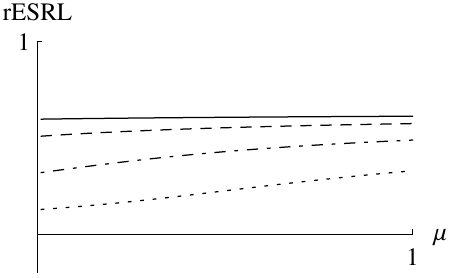}&
\includegraphics[scale=.98]{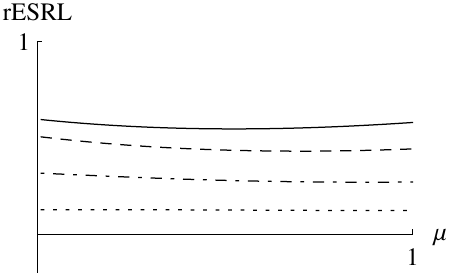}\\
\includegraphics[scale=.98]{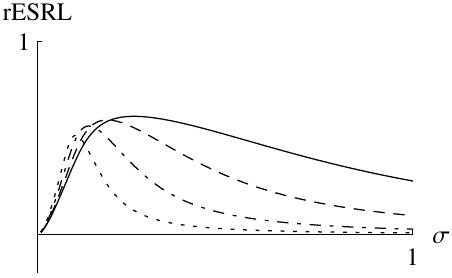}&
\includegraphics[scale=.98]{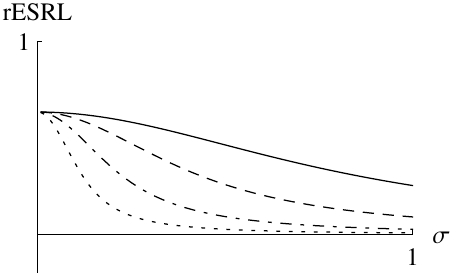}&
\includegraphics[scale=.98]{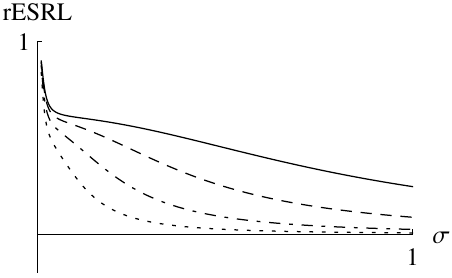}\\
\bottomrule
\end{tabular}
\caption{Dependence of the relative equivalent safe rate loss ($\mathrm{rESRL}$) on $\mu$ and $\sigma$ for $\alpha = 0.1$ (dotted), $\alpha = 0.2$ (dot-dashed), $\alpha = 0.4$ (dashed), and $\alpha =0.8$ (solid). The setup is $T=1$, $G(t)=1-\exp(-t)$, and $\phi'(t) = \alpha$. The parameters are $\sigma =0.2$ (top row) and $\mu = 0.1$ (bottom row).}
\label{table:illustrations:rESRL}
\end{center}
\end{table}

By Theorem~\ref{thm:certainty equivalent}, the addition of a single jump component to the Black--Scholes model reduces the certainty equivalent of trading in the market. We aim to analyze the influence of the model parameters on this welfare loss. A natural quantity to compare different markets is the \emph{equivalent safe rate}. If $\mathrm{CE}$ denotes the certainty equivalent of trading in some market with initial capital $x$ and time horizon $T$, then the \emph{equivalent safe rate} is defined as the unique solution $r:=\mathrm{ESR}$ to the equation $x e^{rT} = \mathrm{CE}$. In other words, the investor is indifferent between trading in this market and receiving a safe annualized return $r$ on his initial capital.\footnote{In a different setting, \cite{GerholdGuasoniMuhleKarbeSchachermayer2013} define the equivalent safe rate slightly differently: they look at the ``long-run'' equivalent safe rate, i.e., the limit as $T \uparrow \infty$.}

Let $\mathrm{CE}^{\mathrm{BS}} = x \exp\left ( \frac{\mu^2}{2 p \sigma^2}T  \right )$ denote the certainty equivalent of trading in a Black--Scholes market. The corresponding \emph{equivalent safe rate} is then given by
\begin{align*}
\mathrm{ESR}^{\mathrm{BS}}
&= \frac{1}{T}\log\left ( \mathrm{CE}^{\mathrm{BS}}/x \right )
= \frac{\mu^2}{2 p \sigma^2}.
\end{align*}
Denoting the certainty equivalent of trading in our market given in \eqref{eqn:thm:certainty equivalent:log utility} and \eqref{eqn:thm:certainty equivalent:power utility} by $\mathrm{CE}$, the corresponding equivalent safe rate is given by
\begin{align}
\label{eqn:ESR}
\mathrm{ESR}
&= \frac{1}{T} \log \left (\mathrm{CE}/x\right )\notag\\
&= \mathrm{ESR}^{\mathrm{BS}} - 
\begin{cases}
\frac{p}{1-p} \frac{1}{T} \log m(0,\hat y(0),p) & \text{if } p\neq 1,\\
\frac{1}{T} \int_0^T \left( \frac{\phi'(u)^2\hat y(u)^2}{2 \sigma^2 \kappa^G(u)} + \log(1 + \hat y (u)) - \frac{\hat y(u)}{1 + \hat y (u)}\right)G'(u) \dd u & \text{if } p= 1.
\end{cases}
\end{align}
In order to improve the comparability over different sets of parameters, we consider the \emph{relative equivalent safe rate loss} $\mathrm{rESRL} = 1 - \frac{\mathrm{ESR}}{\mathrm{ESR}^{\mathrm{BS}}}$ below; it is a relative measure for the incurred losses of trading in our market compared to trading in a Black--Scholes market. It follows from \eqref{eqn:ESR} that
\begin{align*}
\mathrm{rESRL}
&= \frac{2 \sigma^2}{\mu^2}\times
\begin{cases}
\frac{p^2}{1-p} \frac{1}{T} \log m(0,\hat y(0),p) & \text{if } p\neq 1,\\
\frac{1}{T}\int_0^T \left( \frac{\phi'(u)^2\hat y(u)^2}{2 \sigma^2 \kappa^G(u)} + \log(1 + \hat y (u)) - \frac{\hat y(u)}{1 + \hat y (u)}\right)G'(u) \dd u & \text{if } p= 1.
\end{cases}
\end{align*}

Table~\ref{table:illustrations:rESRL} illustrates the dependence of the $\mathrm{rESRL}$ on the model parameters $\mu$ and $\sigma$ as well as on the parameter $\alpha$ describing the relative jump size of the stock price process $S$. In the case $p=4$, the $\mathrm{rESRL}$ is increasing in $\alpha$ and decreasing in $\sigma$ while being almost constant in $\mu$. This is because the diffusive part becomes more and more dominant against the jump part with decreasing $\alpha$ or increasing $\sigma$, so that our model resembles more and more the Black--Scholes model.

In the case $p=0.25$, the dependencies are much less clear. On the one hand, if $\mu$ is sufficiently small and/or $\sigma$ is sufficiently large, then the $\mathrm{rESRL}$ is increasing in $\alpha$. The reason is that in this case, as observed above (cf.~Table~\ref{table:illustrations:dependence on mu sigma alpha:p low}), the optimal strategy does not involve short selling. Therefore, the higher $\alpha$, the higher the losses when the bubble bursts, so that the investor is better off with small jump sizes; this means that the $\mathrm{rESRL}$ is increasing in $\alpha$.

On the other hand, if $\mu$ is high enough and/or $\sigma$ is low enough, so that the optimal strategy involves a significant short position for a significant amount of time, then the investor's wealth is likely to increase when the bubble bursts. Under these circumstances, the investor prefers larger jump sizes; in other words, the $\mathrm{rESRL}$ is decreasing in $\alpha$.

An interesting observation is that for small $\sigma$, i.e., when the jump part dominates the diffusive part, investors with a small relative risk aversion lose only a small fraction of their $\mathrm{ESR}$ compared to an investment into a Black--Scholes market. On the contrary, investors with high relative risk aversion face huge losses in their $\mathrm{ESR}$ for small $\sigma$. This is due to short selling opportunities for investors with low relative risk aversion; cf.~the discussion after Theorem~\ref{thm:decomposition}.

\appendix
\section{Change of filtration}
\label{sec:change of filtration}

Recall from Section~\ref{sec:model class} that the (raw) filtrations $\FF^W = (\cF^W_t)_{t \in [0,T]}$, $\FF^\gamma = (\cF^\gamma_t)_{t \in [0,T]}$, and $\FF = (\cF_t)_{t \in [0,T]}$ are defined by $\cF^W_t = \sigma\left(W_u: 0 \leq u \leq t \right)$, $\cF^\gamma_t = \sigma\left(\1_{\lbrace\gamma \leq u \rbrace}: 0 \leq u \leq t \right)$, and $\cF_t = \sigma\left(\cF^W_t, \cF^\gamma_t\right)$, and that $\FF^W$ and $\FF^\gamma$ are independent under $P$. The key message of the following technical result is that (local) $\FF^\gamma$-martingales are (local) $\FF$-martingales not only under $P$ but also under certain equivalent measures $Q \approx P$, under which $\FF^W$ and $\FF^\gamma$ are no longer independent.

\begin{lemma}
\label{lem:filtration}
Let the function $k: [0, T]^2 \to \RR$ be of the form $k(t,v) = \hat k(t) + \check k(t) \1_{\{t \leq v, t < T\}}$, where $\hat k, \check k \in L^2[0, T]$. Set $Y^1 := \cE\left(\int_0^\cdot k(u, \gamma) \dd W_u \right)$ and let $Z^2$~be a positive $\FF^\gamma$-martingale with $Z^2_0 = 1$. 
\begin{enumerate}
\item Let  $Y^2$ be an $\FF^\gamma$-adapted càdlàg process.
\begin{enumerate}
\item The following are equivalent:
\begin{itemize}
\item [] $Y^2$ is an $\FF^\gamma$-martingale;
\item [] $Y^2$ is an $\FF$-martingale;
\item [] $Y^1 Y^2$ is an $\FF$-martingale.
\end{itemize}
\item If $Y^2$ is a local $\FF^\gamma$-martingale, then $Y^2$ and $Y^1 Y^2$ are local $\FF$-martingales.
\end{enumerate}
\item Define  $Q^\gamma, Q \approx P$ on $\cF_T$ by $\frac{\diff Q^\gamma}{\diff P} = Z^2_T$ and $\frac{\diff Q}{\diff P} = Y^1_T Z^2_T$.\footnote{Note that (a)~(i) with $Y^2 := Z^2$ shows that $Y^1 Z^2$ is a positive $\FF$-martingale with $Y^1_0 Z^2_0 = 1$.} Let $X^{2, Q}$ be an $\cF^\gamma_T$-measurable random variable and $Y^{2,Q}$ an $\FF^\gamma$-adapted càdlàg process.
\begin{enumerate}
\item $X^{2, Q}$ is $Q$-integrable if and only if it is $Q^\gamma$-integrable, and in this case,
\begin{align}
\label{eqn:lem:filtration:Q}
\cEX[Q]{X^{2,Q}}{\cF_s}
&= \cEX[Q^\gamma]{X^{2,Q}}{\cF^\gamma_s} \;\; P \text{-a.s.}, \quad s \in [0, T].
\end{align}
\item $Y^{2,Q}$ is a (square-integrable) $(Q^\gamma,\FF^\gamma)$-martingale if and only if it is a (square-integrable) $(Q,\FF)$-martingale.
\item  If $Y^{2,Q}$ is a local $(Q^\gamma,\FF^\gamma)$-martingale, then it is also a local $(Q,\FF)$-martingale. 
\end{enumerate}
\end{enumerate}
\end{lemma}

\begin{proof}
First, we show that an $\cF^\gamma_T$-measurable random variable $X^2$ is integrable if and only if $Y^1_T X^2$ is so, and in this case,
\begin{align}
\label{eqn:lem:filtration:Bayes P}
\cEX{Y^1_T X^{2}}{\cF_s}
&= Y^1_s \cEX{X^{2}}{\cF^\gamma_s} \;\; P \text{-a.s.}, \quad s \in [0, T].
\end{align}
By linearity, we may assume that $X^2$ is nonnegative. Then the first assertion follows from \eqref{eqn:lem:filtration:Bayes P} for $s = 0$. To establish \eqref{eqn:lem:filtration:Bayes P}, fix $s \in [0, T]$ and set $\cC_s := \{C = C^W \cap \{\gamma \leq u\}: C^W \in \cF^W_s, u \leq s\}$. Then $\cC_s $ is an intersection-closed generator of $\cF_s$, and by the $\cF_s$-measurability of both sides of \eqref{eqn:lem:filtration:Bayes P}, the positivity of $Y^1_s$, and a monotone class argument, it suffices to show that 
\begin{align}
\label{eqn:lem:filtration:00}
\EX{\frac{Y^1_T}{Y^1_s} X^2 \1_C }
&= \EX{\cEX{X^2}{\cF^\gamma_s}\1_C} \quad \text{for all } C \in \cC_s \cup \{\Omega\}.
\end{align}
To establish \eqref{eqn:lem:filtration:00}, for fixed $v \in [0, T]$, set $Y^{1,v} = \cE\left(\int_0^\cdot k(u, v) \dd W_u \right)$. By the assumption on $k$ and Novikov's condition, each $Y^{1,v}$ is a positive $\FF$-martingale with $Y^{1,v}_0 = 1$ and hence satisfies
\begin{align}
\label{eqn:lem:filtration:01}
\EX{\frac{Y^{1,v}_T}{Y^{1,v}_s} \1_A}
&= \EX{\1_A}, \quad s \in [0, T], A \in \cF_s.
\end{align}
Moreover, by the independence of $\cF^\gamma_T = \sigma(\gamma)$ and $W$, a monotone class argument, and \eqref{eqn:lem:filtration:01},
\begin{align}
\label{eqn:lem:filtration:02}
\cEX{\frac{Y^1_T}{Y^1_s}}{\cF^\gamma_T }
&= \left.\EX{\frac{Y^{1, v}_T}{Y^{1,v}_s}}\right\vert_{v = \gamma}
= 1 \;\; \as{P}, \quad s \in [0, T].
\end{align}
Now, \eqref{eqn:lem:filtration:00} for $C = \Omega$ follows from the $\cF^\gamma_T$-measurability of $X^2$ and \eqref{eqn:lem:filtration:02} via
\begin{align*}
\EX{\frac{Y^1_T}{Y^1_s} X^2}
&= \EX{\cEX{\frac{Y^1_T}{Y^1_s}}{\cF^\gamma_T}X^2}
= \EX{X^2}
= \EX{\cEX{X^2}{\cF^\gamma_s}}.
\end{align*}
If $C = C^W \cap \{\gamma \leq u\}$, where $C^W \in \cF^W_s$ and $u \leq s$,  then $Y^1_T/Y^1_s = Y^{1,0}_T/Y^{1,0}_s$ on $C$ since $k(t,v) = \hat k(t) = k(t,0)$ for $t > v$. Moreover, $Y^{1,0}_T/Y^{1,0}_s \1_{C^W}$ is $\cF^W_T$-measurable and $X^2 \1_{\{\gamma \leq u\}}$ and $\cEX{X^2}{\cF^\gamma_t} \1_{\{\gamma \leq u\}}$ are $\cF^\gamma_T$-measurable. This, the independence of $\cF^W_T$ and $\cF^\gamma_T$, and \eqref{eqn:lem:filtration:01} for $A = C^W$ yield
\begin{align*}
\EX{\frac{Y^1_T}{Y^1_s} X^2 \1_C }
&= \EX{\frac{Y^{1,0}_T}{Y^{1,0}_s} X^2 \1_C }
= \EX{\frac{Y^{1,0}_T}{Y^{1,0}_s} \1_{C^W}} \EX{X^2 \1_{\{\gamma \leq u\}}}\\
&= \EX{\1_{C^W}}\EX{\cEX{X^2}{\cF^\gamma_s} \1_{\{\gamma \leq u\}}}
= \EX{\cEX{X^2}{\cF^\gamma_s} \1_{C^W} \1_{\{\gamma \leq u\}}}\\
&= \EX{\cEX{X^2}{\cF^\gamma_s}\1_C}.
\end{align*}

Second, we establish (a). By the first part of the proof, $X^2 := Y^2_T$ is integrable if and only if $Y^1_T Y^2_T$ is so, and in this case,
\begin{align}
\label{eqn:lem:filtration:P martingale}
\cEX{Y^1_T Y^2_T}{\cF_s}
&= Y^1_s \cEX{Y^2_T}{\cF^\gamma_s} \;\; P \text{-a.s.}, \quad s \in [0, T].
\end{align}
As $Y^1_s > 0$ $\as{P}$, \eqref{eqn:lem:filtration:P martingale} shows that $Y^2$ is an $\FF^\gamma$-martingale if and only if $Y^1 Y^2$ is an $\FF$-martingale, and for $Y^1 \equiv 1$ (i.e., for $k \equiv 0$), this implies that $Y^2$ is an $\FF^\gamma$-martingale if and only if it is an $\FF$-martingale. So we have (i). To establish (ii), let $\tau$ be a $[0, T]$-valued $\FF^\gamma$- (and a fortiori $\FF$-)stopping time. Then by the $\FF$-martingale property of $Y^1$ (which follows from Novikov's condition and the assumptions on $k$) and \eqref{eqn:lem:filtration:Bayes P} for $X^2 = \vert Y^2_\tau \vert$ and $s = 0$,
\begin{align*}
\EX{Y^1_\tau \vert Y^2_\tau\vert}
&= \EX{\cEX{Y^1_T}{\cF_\tau}  \vert Y^2_\tau \vert}
= \EX{\cEX{Y^1_T \vert Y^2_\tau \vert}{\cF_\tau}}
= \EX{Y^1_T \vert Y^2_\tau \vert}
= \EX{\vert Y^2_\tau \vert}.
\end{align*}
This implies that $Y^1_\tau Y^2_\tau$ is integrable if and only if $Y^2_\tau$ is so. Now, if the stopped process $(Y^2)^\tau$ is an $\FF^\gamma$-martingale, by the $\FF$-martingale property of $Y^1$, the $\FF^\gamma$-martingale property of $(Y^2)^\tau$, and \eqref{eqn:lem:filtration:Bayes P} for $X := (Y^2)^\tau_T$, for $s \in [0, T]$,
\begin{align*}
\cEX{Y^1_\tau Y^2_\tau  \1_{\{\tau > s\}}}{\cF_s}
&= \cEX{\cEX{Y^1_T}{\cF_{\tau \vee s}} Y^2_\tau  \1_{\{\tau > s\}}}{\cF_s}
= \cEX{\cEX{Y^1_T Y^2_\tau  \1_{\{\tau > s\}}}{\cF_{\tau \vee s}}}{\cF_s}\\
&= \cEX{Y^1_T (Y^2)^\tau_T}{\cF_s} \1_{\{\tau > s\}}
= Y^1_s \cEX{(Y^{2})^\tau_T }{\cF^\gamma_s} \1_{\{\tau > s\}}\\
&= (Y^1)^\tau_s (Y^2)^\tau_s \1_{\{\tau > s\}} \;\; \as{P}
\end{align*}
Thus, $(Y^1)^\tau (Y^2)^\tau$ is an $\FF$-martingale because
\begin{align*}
\cEX{(Y^1)^\tau_T (Y^2)^\tau_T}{\cF_s}
&= Y^1_\tau Y^2_\tau \1_{\{\tau \leq s\}} + \cEX{Y^1_\tau Y^2_\tau  \1_{\{\tau > s\}}}{\cF_s}
= (Y^1)^\tau_s (Y^2)^\tau_s \;\; \as{P}
\end{align*}
For  $Y^1 \equiv 1$ (i.e., for $k \equiv 0$), this also implies that $(Y^2)^\tau$ is an $\FF$-martingale. So if $(\tau_n)_{n \in \NN}$ is a localizing sequence for $Y^2$ in $\FF^\gamma$, it is also a localizing sequence for $Y^2$ and $Y^1 Y^2$ in $\FF$, and we have (ii).

Finally, we establish (b). For (i), set $X^2 = Z^2_T X^{2, Q}$. Then by Bayes' theorem, \eqref{eqn:lem:filtration:Bayes P} for $s = 0$, and again Bayes' theorem,
\begin{align*}
\EX[Q]{\vert X^{2, Q} \vert}
&= \EX[P]{Y^1_T Z^2_T \vert X^{2, Q} \vert}
= \EX[P]{Y^1_T \vert X^2 \vert}
= \EX[P]{\vert X^2 \vert}
= \EX[P]{ Z^2_T \vert X^{2, Q} \vert}
= \EX[Q^\gamma]{\vert X^{2, Q} \vert}\!,
\end{align*}
which shows that $X^{2, Q}$ is $Q$-integrable if and only if it is $Q^\gamma$-integrable. Now, the same argument yields \eqref{eqn:lem:filtration:Q} using \eqref{eqn:lem:filtration:Bayes P} for general $s \in [0, T]$. 

For (ii) and (iii), set $Y^2 = Z^2 Y^{2, Q}$. Then by Bayes' theorem, $Y^{2, Q}$ is a (local) $(Q^\gamma, \FF^\gamma)$-martingale if and only if  $Y^2$ is a (local) $(P, \FF^\gamma)$-martingale. Likewise by Bayes' theorem, $Y^{2, Q}$ is a (local) $(Q, \FF)$-martingale if and only if $Y^1 Y^2$ is a (local) $(P, \FF)$-martingale. Now, (ii) and (iii) follow from (a) (i) and (ii) using also the fact that $(Y^{2,Q}_T)^2$ is $Q$-integrable if and only if it is $Q^\gamma$-integrable; this follows from \eqref{eqn:lem:filtration:Q} for $s = 0$ and $X^{2, Q} = (Y^{2,Q}_T)^2$ using the fact that a martingale on a finite time horizon is square-integrable if and only if it is square-integrable at the final time.
\end{proof}

\section{Analytic results}
\label{sec:analytic results}

The main objective of this section is to show the existence and uniqueness of a solution to the integral equation \eqref{eqn:thm:main result:integral equation}. We first need several preparatory results.

\paragraph{An existence result for ODEs.}

Let $\ul{y} \in C[0,T)$ and $U := \{(t, y) \in [0, T) \times \RR : y > \ul{y}(t)\}$. Let $f: U \to \RR$ be a continuous function that is locally Lipschitz in its second variable. We consider the ordinary differential equation (ODE)
\begin{align}
\label{eqn:general ODE}
y'(t)
&= f(t, y(t)), \quad t \in [0, T).
\end{align}
A function $y \in C^1[0, T)$ with $y > \ul{y}$ is called a \emph{backward upper (lower) solution} to \eqref{eqn:general ODE}
if
\begin{align*}
y'(t)
&\leq \; (\geq) \; f(t, y(t)), \quad t \in [0, T).
\end{align*}
The function $y$ is called a \emph{solution} to \eqref{eqn:general ODE} if it is both a backward upper and a backward lower solution.

\begin{remark}
We define backward upper and lower solution without an initial condition. Moreover, note that what we call backward upper and lower solutions is called \emph{upper and lower solution to the left} in \cite{Walter1998}. Moreover, in \cite{Walter1998} strict (as opposed to weak) inequalities are considered. But as we require $f$ to be locally Lipschitz continuous in its second variable, all results hold also for the weak inequalities (see \cite[Corollary~VIII.9]{Walter1998}).
\end{remark}

The following result gives the existence of a solution to the ODE~\eqref{eqn:general ODE} via the existence of a backward lower and a backward upper solution. The proof for $U = [0, \infty) \times \RR$ can be found in \cite[Theorem and Remark~XIII.9]{Walter1998}, and it is straightforward to check that the argument carries over to our setting.

\begin{lemma}
\label{lem:general existence result}
Let $y_*, y^*  \in C^1[0, T)$ with $y_* \leq  y^*$. Suppose that $y_*$ is a backward lower and $y^*$ a backward upper solution to \eqref{eqn:general ODE}. Then there exists a solution $y \in C^1[0, T)$ to \eqref{eqn:general ODE} with $y_* \leq y \leq y^*$.
\end{lemma}

\paragraph{Properties of the auxiliary functions.}

We collect some analytic properties of the auxiliary functions $a$, $b$, $m$, and $n$ defined in \eqref{eqn:a}, \eqref{eqn:b}, \eqref{eqn:m}, and \eqref{eqn:n}. If there is no danger of confusion, we drop the dependence on $p$ in the notation. It is easy to check that $a, b, m, n \in C([0, T) \times [-1, \infty) \times (0, \infty)) \cap C^{1,2,1}([0, T) \times (-1, \infty) \times (0, \infty))$. For further reference, we note the straightforward identities
\begin{align}
 \label{eqn:da/dy}
\frac{\partial}{\partial y} a(t,y)
&= \frac{1}{p \sigma^2} \frac{\phi'(t)^2}{\kappa^G(t)}
\geq 0,\\
\label{eqn:dm/dy}
\frac{\partial}{\partial y} m(t,y)
&= (1+y)^{\frac{1}{p}} \left(\frac{1}{p} \frac{a(t,y,p)}{1+y} + \frac{\partial}{\partial y} a(t,y,p)\right)
\geq \frac{1}{p} \frac{m(t,y,p)}{1+y},\\
\label{eqn:dn/dy}
\frac{\partial}{\partial y} n(t,y)
&= \kappa^G(t) \left(\frac{1}{p}a(t,y,p) + (1+y) \frac{\partial}{\partial y} a(t,y,p)\right)
\geq \frac{1}{p} \kappa^G(t) a(t,y,p),\\
\label{eqn:n identity}
n(t,y)
&=  -\frac{1 - p}{2 p^2 \sigma^2} \mu^2 + \frac{1 - p}{2 p^2 \sigma^2} \left(\phi'(t) y - \mu\right)^2 + \frac{1}{p} \kappa^G(t)(b(t,y,1) -1).
\end{align}

In view of the integral equation \eqref{eqn:thm:main result:integral equation}, we are interested in the domain where the function $m$ is positive. To this end, define the function $\ul{y}: [0, T) \to [-1, \infty)$ by
\begin{align}
\label{eqn:yunderline}
\ul{y}(t)
&:=
\begin{cases}
-1 &\text{if } \phi'(t) = 0,\\
\max\left(-1, \frac{\mu}{\phi'(t)} - p \sigma^2 \frac{\kappa^G(t)}{\phi'(t)^2}\right) & \text{if } \phi'(t) > 0.
\end{cases}
\end{align}
Using that $\kappa^G$ is continuous and positive on $[0, T)$, it is not difficult to check that $\ul{y} \in C[0, T)$. Set
\begin{align}
\label{eqn:U}
U
&= \{(t, y) \in [0, T) \times \RR: y > \ul{y}(t)\}.
\end{align}
Then by the definition of $\ul{y}$, \eqref{eqn:dm/dy}, and \eqref{eqn:dn/dy}, 
\begin{align}
\label{eqn:auxiliary functions >0} 
a(t,y), m(t,y), \frac{\partial}{\partial y} m(t,y), \frac{\partial}{\partial y} n(t,y)
&> 0, \quad (t, y) \in U. 
\end{align}

\paragraph{An implicit function result.}

The following inverse-function-type result is the cornerstone of the subsequent analysis. In particular, it is used in Theorem~\ref{thm:integral equation} to construct backward upper and backward lower solutions for the ODE \eqref{eqn:general ODE}. Recall the definition of $\ul{y}$ in \eqref{eqn:yunderline}.
\begin{lemma}
\label{lem:implicit function}
Fix $p \in (0, \infty)$. Let $f \in C^1[0, T)$ with $f(t) > 0$, $t \in [0, T)$. Then there exists a unique function $y \in C^1[0, T)$ with $y > \ul{y}$ such that
\begin{align}
\label{eqn:lem:implicit function}
m(t, y(t))
&= f(t).
\end{align}
Moreover, if $\lim_{t \uparrow \uparrow T} f(t) = 1$, then there exist constants $\epsilon \in (0, 1]$ and $C \geq 1$ such that
\begin{align}
\label{eqn:lem:implicit function:strict local martingale inequalities}
\epsilon
&\leq 1 + y(t)
\leq C + \frac{C}{\phi'(t)} \1_{\{\kappa^G(t)  < C \phi'(t)\}}, \quad t \in [0, T).
\end{align}
In this case, if in addition $\int_0^T \left\vert \phi'(u) y(u)\right\vert  \dd u  < \infty$, then 
\begin{align}
\label{eqn:lem:implicit function:integrability condition}
\int_0^T \1_{\lbrace \Delta G(T) > 0 \rbrace} \kappa^G(u) (1+ y(u)) \dd u
&< \infty.
\end{align}
\end{lemma}

\begin{proof}
First, for fixed $t \in [0, T)$, by \eqref{eqn:m} and \eqref{eqn:dm/dy}, $y \mapsto m(t, y)$ is continuous and increasing\footnote{We emphasize that we use qualifiers like ``increasing'', ``decreasing'', ``positive'', ``negative'' in the \emph{strict} sense; the corresponding wide-sense notions are ``nondecreasing'', ``nonincreasing'', ``nonnegative'', ``nonpositive''.} on $[\ul{y}(t), \infty)$ with $m(t,\ul{y}(t)) = 0$ and $\lim_{y \to \infty} m(t,y) = +\infty$. Thus, there exists a unique function $y : [0, T) \to \RR$ with $y > \ul{y}$ satisfying \eqref{eqn:lem:implicit function}. Moreover, $y \in C^1(0, T)$ by the implicit function theorem. 

Second, for fixed $t \in [0, T)$, we claim that
\begin{align}
\label{eqn:lem:implicit function:pf:01} 
y(t)
&\leq (2 f(t))^p
&&\text{if } \phi'(t) \leq \frac{p \sigma^2}{2 \mu} \kappa^G(t),\\
\label{eqn:lem:implicit function:pf:02}
y(t)
&\leq  \max\left(f(t)^p, \frac{\mu}{\phi'(t)}\right)
&&\text{if } \phi'(t) > 0.
\end{align}
Indeed, fix $t \in [0, T)$. If $\phi'(t) \leq \frac{p \sigma^2}{2 \mu} \kappa^G(t)$, then $a(t,0) = 1 - \frac{\mu}{p \sigma^2} \frac{\phi'(t)}{\kappa^G(t)} \geq \frac{1}{2}$. Seeking a contradiction, suppose that $y(t) > (2 f(t))^p > 0$. Then by the definitions of $m$ and $y(t)$ and the monotonicity of $a$ in the second variable, 
\begin{align*}
f(t)
&= m(t,y(t))
= (1+  y(t))^{\frac{1}{p}} a(t, y(t))
> (y(t))^{\frac{1}{p}} a(t,0)
\geq f(t),
\end{align*}
which is absurd. If $\phi'(t) >0$, then $a\left(t, \frac{\mu}{\phi'(t)}\right) = 1$, and \eqref{eqn:lem:implicit function:pf:02} follows from a similar argument.

Third, by the implicit function theorem and \eqref{eqn:dm/dy}, for $t \in (0, T)$,
\begin{align*}
\left\vert y'(t)\right\vert
&= \left\vert\frac{f'(t) - \frac{\partial}{\partial t} m(t, y(t))}{\frac{\partial}{\partial y} m(t, y(t))}\right\vert
\leq  p(1+y(t)) \frac{\left\vert f'(t) - \frac{\partial}{\partial t} m(t, y(t))\right\vert}{m(t, y(t))}\\
&= p(1+y(t)) \frac{\left\vert f'(t) - \frac{\partial}{\partial t} m(t, y(t))\right\vert}{f(t)}.
\end{align*}
Now, fix $t_0 \in (0, T)$ and let $C > 0$ be such that $y(t) \leq C$ for all $t \in [0, t_0]$. (This is possible by \eqref{eqn:lem:implicit function:pf:01}, \eqref{eqn:lem:implicit function:pf:02} and the facts that $f$ is continuous and $\kappa^G$ is continuous and positive). Then the positivity and continuity of $f$ in $[0, t_0]$, the continuity of $f'$ in $[0, t_0]$, and the continuity of $\frac{\partial}{\partial t} m(t, y)$ in $[0, t_0] \times [-1, C]$ together with the fact that $-1 \leq y(t) \leq C$ for $t \in [0, t_0]$ show that $y'$ is uniformly bounded in $(0, t_0]$. Moreover, by the fundamental theorem of calculus and the fact that $y \in C^1(0,T)$,
\begin{align*}
y(t)
&= y(t_0) - \int_t^{t_0} y'(u) \dd u, \quad t \in (0, t_0].
\end{align*}
Thus, by dominated convergence, $\lim_{t \downarrow \downarrow 0} y(t)$ exists in $\RR$. The continuity of $f$ and $m$ and \eqref{eqn:lem:implicit function} give
\begin{align*}
m(0, y(0))
&= f(0)
= \lim_{t \downarrow \downarrow 0} f(t)
= \lim_{t \downarrow \downarrow 0} m(t, y(t))
= m(0, \lim_{t \downarrow \downarrow 0} y(t)),
\end{align*}
and so by the uniqueness of $y$ on $[0, T)$, $\lim_{t \downarrow \downarrow 0} y(t) = y(0) > \ul{y}(0) \geq -1$. This together with the continuity of $\frac{\partial}{\partial y}m$ and $\frac{\partial}{\partial t}m$ on $[0, T) \times (-1, \infty)$, the continuity of $f'$ on $[0, T)$, and the identity $y'(t)= \left(f'(t) - \frac{\partial}{\partial t} m(t, y(t))\right)/\frac{\partial}{\partial y} m(t, y(t))$ for $t \in (0, T)$ (by the implicit function theorem) shows that the limit $\lim_{t \downarrow \downarrow 0} y'(t)$ exists in $\RR$. So $y \in C^1[0, T)$.

Fourth, assume $\lim_{t \uparrow \uparrow T} f(t) = 1$. Set
\begin{align}
\label{eqn:lem:implicit function:pf:C}
C
&:= 1+ \max\left(\sup_{t \in [0, T)} (2 f(t))^p, \mu, \frac{2 \mu}{p \sigma^2}\right). 
\end{align}
Then $1 \leq C < \infty$. Fix $t \in [0, T)$. If $\kappa^G(t) \geq C \phi'(t)$, then $\phi' (t) \leq \frac{1}{C} \kappa^G(t) \leq \frac{p \sigma^ 2}{2 \mu} \kappa^G(t)$, and so $1 + y(t) \leq C$ by \eqref{eqn:lem:implicit function:pf:01} and the definition of $C$.
Otherwise, if $\kappa^G(t) < C \phi'(t)$, then $\phi'(t) > 0$, and so $1 + y(t) \leq C +\frac{C}{\phi'(t)}$ by \eqref{eqn:lem:implicit function:pf:02} and the definition of $C$. For the left inequality in \eqref{eqn:lem:implicit function:strict local martingale inequalities}, by the continuity of $y$ in $[0, T)$ and the fact that $y > \ul y \geq -1$ on $[0, T)$, it suffices to show that $\liminf_{t \uparrow \uparrow T} y(t) > -1$.  Seeking a contradiction, suppose there is a sequence $(t_n)_{n \in \NN} \subset [0, T)$ increasing to $T$ such that $\lim_{n \to \infty} y(t_n) = -1$. Passing to a subsequence if necessary, we may assume that $y(t_n) \leq 0$ for all $n \in \NN$. As $\phi' \geq 0$ by  \eqref{eqn:standing assumption}, the definition of $a$ in \eqref{eqn:a} gives $a(t_n,y(t_n)) \leq 1$ for all $n \in \NN$. Now, using the definition of $m$ in \eqref{eqn:m}, we arrive at the contradiction
\begin{align*}
1
&= \lim_{n \to \infty} f(t_n)
= \lim_{n \to \infty} m(t_n, y(t_n))
\leq \limsup_{n \to \infty} (1 + y(t_n))^ {\frac{1}{p}}
= 0.
\end{align*}

Finally, assume that in addition $\Delta G(T) > 0$ and $\int_0^T \left\vert \phi'(u) y(u)\right\vert\dd u < \infty$. Then by \eqref{eqn:kappa},
\begin{align*}
\int_0^T \kappa^G(u) \dd u
&= -\log(\Delta G(T))
< \infty.
\end{align*}
Define $C$ as in \eqref{eqn:lem:implicit function:pf:C}, and set $A := \{u \in [0, T): \phi'(u)  \leq \frac{p \sigma^2}{2 \mu} \kappa^G(u)\}$. Then $y(u) \leq C$ for $u \in A$ by \eqref{eqn:lem:implicit function:pf:01}, and $\kappa^G(u) < \frac{2 \mu}{p \sigma^2} \phi'(u)$ for $u \in A^c$. This together with the above yields \eqref{eqn:lem:implicit function:integrability condition} via
\begin{align*}
\int_0^T (1 +  y(u))\kappa^G(u) \dd u
&= \int_0^T \left(\kappa^G(u) + \1_{A}(u) y (u) \kappa^G(u)  + \1_{A^c}(u)  y(u) \kappa^G(u) \right) \dd u \\
&\leq \int_0^T \left((1 + C) \kappa^G(u)  + \frac{2 \mu}{p \sigma^2} \phi'(u)\vert y(u)\vert \right) \dd u
< \infty.\qedhere
\end{align*}
\end{proof}

\begin{corollary}
\label{cor:implicit function}
Fix $p \in (0, \infty)$. Let $f, g \in C^1[0, T)$ be such that $g(t) >f(t) > 0$, $t \in [0, T)$, and $\lim_{t \uparrow \uparrow T} g(t) = \lim_{t \uparrow \uparrow T} f(t) = 1$. Let $y^f, y^g \in C^1[0, T)$ with $y^g,  y^f > \ul y$ be the unique functions from Lemma~\ref{lem:implicit function} satisfying $m(t, y^f(t)) = f(t)$ and $m(t, y^g(t)) = g(t)$. Assume that $\Delta G(T) > 0$ and $\limsup_{t \uparrow \uparrow T} G'(t) < \infty$. Then
\begin{align*}
\lim_{t \uparrow \uparrow T} \phi'(t)(y^g(t) - y^f(t))
&= 0.
\end{align*}
\end{corollary}

\begin{proof}
By Lemma~\ref{lem:implicit function}, there are constants $C^g, \epsilon^f > 0$ such that $1 +y^g(t) \leq C^g + \frac{C^g}{\phi'(t)} \1_{\{\phi'(t) > 0\}}$ and $1 +y^f(t) \geq \epsilon^f$, $t \in [0, T)$. As $\Delta G(T) > 0$ and $\limsup_{t \uparrow \uparrow T} G'(t) < \infty$, there exists a constant $C^{\kappa} > 0$ such that $\frac{1}{\kappa^G(t)} = \frac{1- G(t)}{G'(t)} \geq C^{\kappa}$, $t \in [0, T)$. Next, as $f \in C^1[0, T)$, $f > 0$ on $[0, T)$, and $\lim_{t \uparrow \uparrow T} f(t) = 1$, there exists a constant $C^f > 0$ such that $f(t) \geq C^f$, $t \in [0, T)$. Finally, set $C = \min\left(\frac{C^f}{2 C^g}, \frac{(\epsilon^f)^{\frac{1}{p}} C^\kappa}{\sigma^2}\right) > 0$. Fix $t \in [0, T)$. Then $y^g(t) \geq y^f(t)$ by the monotonicity of $m$ in the second variable. Using successively \eqref{eqn:lem:implicit function}, the mean value theorem, \eqref{eqn:dm/dy}, \eqref{eqn:m} and \eqref{eqn:da/dy}, the monotonicity of $m$ in the second variable, \eqref{eqn:lem:implicit function} and the choices of $C^g$, $\epsilon^f$, and $C^\kappa$, the choice of $C^f$, and finally the choice of $C$ (distinguishing the cases $\phi'(t) \geq 1$ and $\phi'(t) < 1$) yields, for some $\tilde y \in [y^f(t), y^g(t)]$,
\begin{align*}
g(t) - f(t)
&= m(t, y^g(t)) - m(t, y^f(t))
= (y^g(t) - y^f(t)) \frac{\partial}{\partial y} m(t, \tilde y) \notag \\
&=(y^g(t) - y^f(t)) (1+\tilde y)^{\frac{1}{p}} \left(\frac{1}{p} \frac{a(t,\tilde y)}{1+\tilde y} + \frac{\partial}{\partial y} a(t,\tilde y)\right) \\
&= \frac{1}{p} (y^g(t) - y^f(t)) \left(\frac{m(t,\tilde y)}{1+\tilde y} + \frac{1}{\sigma^2} (1 + \tilde y)^{\frac{1}{p}} \frac{\phi'(t)^2}{\kappa^G(t)}\right) \\
&\geq \frac{1}{p} (y^g(t) - y^f(t)) \left(\frac{m(t,y^f(t))}{1+y^g(t)} + \frac{1}{\sigma^2} (1 + y^f(t))^{\frac{1}{p}} \frac{\phi'(t)^2}{\kappa^G(t)}\right) \\
&\geq \frac{1}{p} (y^g(t) - y^f(t)) \left(\frac{f(t)}{C^g + \frac{C^g}{\phi'(t)} \1_{\{\phi'(t) > 0\}}} + \frac{(\epsilon^f)^{\frac{1}{p}} C^\kappa}{\sigma^2} \phi'(t)^2 \right) \\
&\geq \frac{1}{p} \phi'(t) (y^g(t) - y^f(t)) \left(\frac{C^f}{C^g \phi'(t)+ C^g } + \frac{(\epsilon^f)^{\frac{1}{p}} C^\kappa}{\sigma^2} \phi'(t) \right) \\
&\geq \frac{C}{p} \phi'(t) (y^g(t) - y^f(t)).
\end{align*}
Now, the claim follows from letting $t \uparrow\uparrow T$.
\end{proof}

\paragraph{Existence and uniqueness of a solution to the integral equation.}

We are now in a position to prove the main existence and uniqueness result for the integral equation \eqref{eqn:thm:main result:integral equation}. Recall the definition of $\ul{y}$ in \eqref{eqn:yunderline}.

\begin{theorem}
\label{thm:integral equation}
Fix $p \in (0, \infty)$. Then there exists a unique solution $\hat y \in C^1[0, T)$ with $\hat y > \ul{y}$ to the integral equation 
\begin{align}
\label{eqn:thm:main result:integral equation:appendix}
m(t,y(t),p)
&= \exp\left( - \int_t^T n(u,y(u),p) \dd u\right),\quad t\in[0,T),
\end{align}
satisfying \eqref{eqn:lem:implicit function:strict local martingale inequalities} and \eqref{eqn:lem:implicit function:integrability condition} (with $y$ replaced by $\hat y$) as well as
\begin{align}
\int_0^T \vert n(u, \hat y(u), p) \vert \dd u
&< \infty
\quad\text{and}\quad
\int_0^T \left(\phi'(u) \hat y(u) \right)^2 \dd u
< \infty.
\label{eqn:thm:integral equation:integrability properties}
\end{align}
Moreover, $y_*  \leq \hat y \leq y^*$ on $[0, T)$, where $y_*, y^* \in C^1[0,T)$ are the unique functions from Lemma~\ref{lem:implicit function} satisfying $y_*, y^* > \ul{y}$ and
\begin{align}
 \label{eqn:thm:integral equation:upper solutions}
m(t, y^*(t),p)
&= 
\begin{cases}
\exp\left(\frac{1- p}{2 p^2 \sigma^2} \mu^2 (T- t)\right) & \text{if } p < 1,\\
1 &\text{if } p \geq 1,
\end{cases}\\
\label{eqn:thm:integral equation:lower solutions}
m(t, y_*(t),p)
&= 
\begin{cases}
1 & \text{if } p < 1,\\
\exp\left(\frac{1- p}{2 p^2 \sigma^2} \mu^2 (T-t)\right) &\text{if } p \geq 1.
\end{cases}
\end{align}
\end{theorem}

Note that \eqref{eqn:lem:implicit function:strict local martingale inequalities} and \eqref{eqn:lem:implicit function:integrability condition} (with $y$ replaced by $\hat y$) as well as \eqref{eqn:thm:integral equation:integrability properties} imply in particular that \eqref{eqn:thm:ELMM} and \eqref{eqn:thm:strict local martingale} (with $y$ replaced by $\hat y$) are fulfilled.

\begin{proof}
First, we transform the integral equation \eqref{eqn:thm:main result:integral equation:appendix} into an ODE. Taking logarithms on both sides of \eqref{eqn:thm:main result:integral equation:appendix} and differentiating shows that a solution $y \in C^1[0,T)$ to \eqref{eqn:thm:main result:integral equation} solves
\begin{align}
\label{eqn:thm:integral equation:pf:log ODE}
\frac{\diff}{\diff t} \log(m(t, y(t)),p)
&= n(t,y(t),p).
\end{align}
An easy calculation using \eqref{eqn:m} and \eqref{eqn:dm/dy} gives
\begin{align*}
\frac{\diff}{\diff t} \log(m(t, y(t),p))
&= \frac{y'(t)\left(\frac{1}{p} \frac{a(t, y(t),p)}{1+ y(t)} + \frac{\partial}{\partial y} a(t, y(t),p)\right) + \frac{\partial}{\partial t} a(t, y(t),p) }{a(t, y(t),p)}.
\end{align*}
Rearranging the terms shows that $y$ solves the ODE 
\begin{align}
\label{eqn:thm:integral equation:pf:ODE}
y'(t)
&= f(t, y(t),p), \quad t \in [0, T),
\end{align}
where the function $f : U \times (0 , \infty) \to (0, \infty)$ is given by
\begin{align*}
f(t,y,p)
&=\frac{a(t,y,p) n(t,y,p) - \frac{\partial}{\partial t}a(t,y,p)}{\frac{1}{p} \frac{a(t,y,p)}{1+y} + \frac{\partial}{\partial y} a(t,y,p)}
\end{align*}
and $U$ is defined in \eqref{eqn:U}. Clearly, $f \in C^{0,1,1}(U \times (0, \infty))$. Note that the positivity of the denominator is ensured by the positivity of $a$ in $U\times(0,\infty)$ by \eqref{eqn:auxiliary functions >0} and \eqref{eqn:da/dy}. Moreover, \eqref{eqn:thm:main result:integral equation:appendix} gives the  implicit ``terminal condition''
\begin{align*}
\lim_{t \uparrow \uparrow T} m(t,y(t),p)
&= 1.
\end{align*}

Second, we establish the uniqueness of $\hat y$. Assume that $\hat y^1,\hat y^2 \in C^1[0, T)$ are two solutions of \eqref{eqn:thm:main result:integral equation:appendix}. Then $\hat y^1, \hat y^2 > \ul y$ and both functions are solutions to the ODE \eqref{eqn:thm:integral equation:pf:ODE}. Assume without loss of generality that $\hat y^2(0) \geq \hat y^1(0)$. As $f$ is locally Lipschitz in the second variable on $U$, it follows from the standard local existence and uniqueness theorem for ODEs that either $\hat y^1 = \hat y^2$ or $\hat y^2 > \hat y^1$. Seeking a contradiction, assume the second case. Then by the strict monotonicity of $m$ and $n$ in the second variable (by \eqref{eqn:auxiliary functions >0}) and the fact that $\hat y^1$ and $\hat y^2$ are solutions to \eqref{eqn:thm:main result:integral equation},
\begin{align*}
m(0, \hat y^2(0))
&> m(0, \hat y^1(0))
= \exp\bigg(-\int_0^T n(u, \hat y^1(u)) \dd u \bigg) \\
&> \exp\bigg(-\int_0^T n(u, \hat y^2(u)) \dd u \bigg)
= m(0, \hat y^2(0)),
\end{align*}
which is absurd. So $\hat y^1 = \hat y^2$.

Third, we use Lemma~\ref{lem:general existence result} to show the existence of a solution to \eqref{eqn:thm:integral equation:pf:ODE}. To this end, we show that $y^*$ and $y_*$ are backward upper and backward lower solutions, respectively. The existence and uniqueness of the functions $y_*$ and $y^*$ satisfying \eqref{eqn:thm:integral equation:upper solutions} and \eqref{eqn:thm:integral equation:lower solutions} follows from Lemma~\ref{lem:implicit function}. Note that $y \in C^1[0, T)$ with $y > \ul{y}$ is a backward upper (lower) solution to \eqref{eqn:thm:integral equation:pf:ODE} if and only if
\begin{align}
\label{eqn:thm:integral equation:pf:criterion}
\frac{\diff}{\diff t} \log(m(t, y(t),p))
&\leq (\geq) \; n(t,y(t),p), \quad t \in [0, T);
\end{align}
this follows from the same rearrangement that led from \eqref{eqn:thm:integral equation:pf:log ODE} to \eqref{eqn:thm:integral equation:pf:ODE} using that $\frac{1}{p} \frac{a}{1+y} + \frac{\partial}{\partial y} a$ and $a$ are positive in $U \times (0, \infty)$ by \eqref{eqn:auxiliary functions >0} and \eqref{eqn:da/dy}.

We only consider the case $p < 1$; the case $p \geq 1$ follows from a similar argument, basically reversing all inequalities. Bernoulli's inequality, \eqref{eqn:b}, and \eqref{eqn:m} yield
\begin{align}
\label{eqn:thm:integral equation:pf:Bernoulli:p<1}
b(t,y,p)
&\leq m(t,y,p)
\leq b(t,y,1)^{1/p}, \quad (t,y) \in U.
\end{align}
To establish that $y^*$ is a backward upper solution, note from \eqref{eqn:thm:integral equation:upper solutions} that $m(t, y^*(t), p) \geq 1$ for $t \in [0, T)$. Thus, $b(t, y^*(t),1) \geq 1$ for $t \in [0, T)$ by \eqref{eqn:thm:integral equation:pf:Bernoulli:p<1}, and so $n(t, y^*(t),p) \geq -\frac{1 - p}{2 p^2 \sigma^2} \mu^2$ for $t \in [0, T)$ by \eqref{eqn:n identity}. Now, taking logarithms in \eqref{eqn:thm:integral equation:upper solutions} and differentiating shows that $y^*$ fulfills \eqref{eqn:thm:integral equation:pf:criterion} with ``$\leq$'', and so $y^*$ is a backward upper solution. To establish that $y_*$ is a backward lower solution, note from \eqref{eqn:thm:integral equation:lower solutions} that $m(t, y_*(t), p) = 1$ for $t \in [0, T)$, and so $b(t, y_*(t),p) \leq 1$ for $t \in [0, T)$ by \eqref{eqn:thm:integral equation:pf:Bernoulli:p<1}. Thus, $n(t, y_*(t),p) \leq 0$ by \eqref{eqn:n}, and the claim follows as above by taking logarithms in \eqref{eqn:thm:integral equation:lower solutions} and differentiating. Clearly, $y_* \leq y^*$ by the monotonicity of $m$ in in the second variable, and $\lim_{t \uparrow \uparrow T} m(t, y_*(t),p) = \lim_{t \uparrow \uparrow T} m(t, y^*(t),p) = 1$ by construction. So by Lemma~\ref{lem:general existence result}, there exists a solution $\hat y \in C^1[0, T)$ of \eqref{eqn:thm:integral equation:pf:ODE} with $y_* \leq \hat y \leq y^*$.

Fourth, $\hat y > \ul y$ because $\hat y  \geq y_* > \ul y$ by construction. The monotonicity of $m$ in the second variable and the fact that $\lim_{t \uparrow \uparrow T} m(t, y^*(t),p) = 1 = \lim_{t \uparrow \uparrow T} m(t, y_*(t),p) $ by \eqref{eqn:thm:integral equation:upper solutions} and \eqref{eqn:thm:integral equation:lower solutions} yield $\lim_{t \uparrow \uparrow T} m(t, \hat y(t),p) = 1$. Moreover, as $\hat y$ satisfies \eqref{eqn:thm:integral equation:pf:log ODE}, the fundamental theorem of calculus shows that there exists a constant $c > 0$ such that $\hat y$ satisfies the integral equation
\begin{align}
\label{eqn:thm:integral equation:pf:integral equation}
m(t,\hat y(t),p)
&= c \exp\left( \int_0^t n(u, \hat y(u),p) \dd u\right), \quad t \in [0, T).
\end{align}
Now, we have to distinguish the cases $p < 1$ and $p  \geq 1$. We only consider the first one; the second one follows from a similar argument, basically reversing all inequalities. So let $p \in (0, 1)$. Then $m(t,\hat y(t),p) \geq m(t,y_*(t),p) = 1$ by the monotonicity of $m$ in the second variable and \eqref{eqn:thm:integral equation:lower solutions}. Thus,  \eqref{eqn:thm:integral equation:pf:Bernoulli:p<1} gives $b(t,\hat y(t),1) \geq 1$, and so $n(t, \hat y(t),p) \geq -\frac{1-p}{2 p^2 \sigma ^2} \mu^2$ by  \eqref{eqn:n identity}. Taking the limit $t \uparrow \uparrow T$ in \eqref{eqn:thm:integral equation:pf:integral equation}, by monotone convergence and the fact that $\lim_{t \uparrow \uparrow T} m(t,\hat y(t),p) = 1$, we deduce that
\begin{align}
\label{eqn:thm:integral equation:pf:constant c}
c
&= \exp\left(- \int_0^T n(u, \hat y(u),p) \dd u\right).
\end{align}
Plugging this back into \eqref{eqn:thm:integral equation:pf:integral equation} shows that $\hat y$ is a solution to \eqref{eqn:thm:main result:integral equation:appendix}. Moreover, as $n(t,\hat y(t), p)$ is bounded from below and $c > 0$, \eqref{eqn:thm:integral equation:pf:constant c} implies that the first condition in \eqref{eqn:thm:integral equation:integrability properties} is satisfied. This together with the representation of $n$ in \eqref{eqn:n identity} and $b(t,\hat y(t),1) \geq 1$ (from above) then also establishes the second condition in \eqref{eqn:thm:integral equation:integrability properties}. Finally, define $\hat f(t)$ by the right-hand side of \eqref{eqn:thm:main result:integral equation:appendix} (with $y$ replaced by $\hat y$). Then $\hat y$ is trivially a solution to $m(t, \hat y(t)) = \hat f(t)$, $t \in [0, T)$, and $\lim_{t \uparrow\uparrow T} \hat f(t) = 1$. Hence, Lemma~\ref{lem:implicit function} gives \eqref{eqn:lem:implicit function:strict local martingale inequalities} and \eqref{eqn:lem:implicit function:integrability condition} for $\hat y$ (note that the condition $\int_0^T \vert \phi'(u) \hat y(u) \vert \dd u < \infty$ follows from \eqref{eqn:thm:integral equation:integrability properties}).
\end{proof}

\section{Verification }
\label{sec:verification}

Here, we collect the technical parts of Steps 2 and 3 of the proof of Theorem~\ref{thm:main result}. The first result identifies the wealth process corresponding to the strategy $\hat\pi$ and shows that it remains positive. The second and third result verify (OC1) and (OC2) for the candidate triplet $(\hat\pi,\hat Q,\hat z)$.

\begin{lemma}
\label{lem:wealth process}
Let $(\hat\pi,\hat Q,\hat z)$ be the triplet defined in (the proof of) Theorem~\ref{thm:main result}. Denote by $W^{\hat Q}$ the $\hat Q$-Brownian motion given by Theorem~\ref{thm:ELMM} (with $y = \hat y)$ and let $\hat H$ be the distribution function of $\gamma$ under $\hat Q$. Define $\hat\xi \in C^1[0,T)$ by
\begin{align}
\label{eqn:xihat}
\hat\xi (t)
&= \exp \left (\int_0^t \phi'(u)\frac{1}{p \sigma^2}(\mu - \phi'(u) \hat y(u))(1+\hat y(u))  \dd u  \right )
\end{align}
and set $\hat X := \cE\left ( \sigma \int_0^\cdot \hat\pi_t \dd W^{\hat Q}_t \right)\mart{{\hat H}}{\hat\xi}$. Then $X^{\hat\pi} = x \hat X$ is the wealth process corresponding to the strategy $\hat \pi$ and initial capital $x$. Moreover, $\mart{\hat H}{\hat\xi}$ and $\hat X$ are positive and thus $\hat\pi$ is admissible.
\end{lemma}

\begin{proof}
For the first claim, it suffices to show that $\hat X$ satisfies the SDE \eqref{eqn:wealth process SDE} with initial condition $\hat X_0 = 1$. Set $M := \cE\left( \sigma \int_0^\cdot \hat\pi_t \dd W^{\hat Q}_t \right)$ and $N := \mart{{\hat H}}{\hat\xi}$ for brevity and note from \eqref{eqn:thm:main result:trading strategy} that $\hat\pi_t = \bar\pi(t,\gamma)$, $t\in[0,T]$, where
\begin{align}
\label{eqn:pibar}
\bar\pi(t, v)
&:= \frac{1}{p \sigma^2} \left ( \mu - \phi'(t)\hat y(t)\1_{\lbrace t \leq v, t < T \rbrace }  \right ), \quad (t,v) \in [0,T]^2.
\end{align}
With this notation, by the definition of $\hat\xi$, we obtain
\begin{align}
\label{eqn:xihat ODE}
\hat\xi'(t)
&= \hat\xi(t) \phi'(t) \bar\pi(t,t) (1 + \hat y(t)), \quad t \in [0,T).
\end{align}
Fix $t \in [0,T]$. By using successively that $M$ is continuous and $N$ is purely discontinuous, that $\Delta \mart{{\hat H}}{\hat\xi}_\gamma = -\frac{\hat\xi'(\gamma)}{\kappa^{\hat H}(\gamma)}\1_{\lbrace \gamma < T\rbrace}$ by the definitions of $\mart{\hat H}{\hat\xi}$ (cf.~\eqref{eqn:MGF}) and $\cA^{\hat H}\hat\xi$ (cf.~\eqref{eqn:AGF}), \eqref{eqn:xihat ODE}, that $\hat\xi(s) = N_{s-}$ and $\bar\pi(s,s) = \bar\pi(s,\gamma) = \hat\pi_s$ on $\lbrace s \leq \gamma \rbrace$, and finally the dynamics of $S$ in \eqref{eqn:thm:ELMM:S Qdynamics} (for $y = \hat y$ etc.),
\begin{align*}
\hat X_t - 1
&= M_t N_t -M_0 N_0 = \int_0^t N_{s-} \dd M_s + \int_0^t M_{s-} \dd N_s\\
&= \sigma \int_0^t N_{s-} M_{s-} \hat\pi_s \dd W^{\hat Q}_s + \int_0^t M_{s-} \dd\mart{{\hat H}}{\hat\xi}_s\\
&= \sigma \int_0^t \hat X_{s-} \hat\pi_s \dd W^{\hat Q}_s
+ \int_0^{t \wedge \gamma} M_{s-} \hat\xi'(s)\dd s
+ M_{\gamma-}\Delta \mart{{\hat H}}{\hat\xi}_\gamma \1_{\lbrace \gamma \leq t\rbrace} \\
&= \sigma \int_0^t \hat X_{s-} \hat\pi_s \,\diff W^{\hat Q}_s
+ \int_0^{t \wedge \gamma}  M_{s-}\hat\xi'(s) \, \diff s
- M_{\gamma-} \frac{\hat\xi'(\gamma)}{\kappa^{\hat H}(\gamma)} \1_{\lbrace \gamma \leq t, \gamma < T\rbrace}\\
&= \sigma \int_0^t \hat X_{s-} \hat\pi_s \,\diff W^{\hat Q}_s
+\int_0^{t \wedge \gamma}  M_{s-}\hat\xi(s) \bar\pi(s,s) \phi'(s)(1 + \hat y(s)) \, \diff s\\
&\qquad-  M_{\gamma-} \hat\xi(\gamma)\bar\pi(\gamma,\gamma) \frac{\phi'(\gamma)(1+\hat y(\gamma))}{\kappa^{\hat H}(\gamma)} \1_{\lbrace \gamma \leq t, \gamma < T \rbrace}\\
&= \sigma \int_0^t \hat X_{s-} \hat\pi_s \,\diff W^{\hat Q}_s
+ \int_0^{t \wedge \gamma}  M_{s-}N_{s-}\hat\pi_s \phi'(s)(1 + \hat y(s)) \, \diff s\\
&\qquad+ M_{\gamma-}N_{\gamma-} \hat\pi_\gamma \Delta \mart{{\hat H}}{\left (\int_0^\cdot \phi'(u)(1+ \hat y(u)) \,\diff u \right )}_\gamma \1_{\lbrace \gamma \leq t\rbrace}\\
&= \int_0^t \hat\pi_s \hat X_{s-} \left [\sigma \dd W^{\hat Q}_s + \diff \mart{{\hat H}}{\left(\int_0^\cdot \phi'(u)(1+ \hat y(u)) \,\diff u \right )}_s\right ]\\
&= \int_0^t \hat\pi_s \hat X_{s-} \frac{\diff S_s}{S_{s-}} \quad \as{P}
\end{align*}

For the second claim, as $M$ and $\hat\xi$ are positive, it suffices to show that also $\cA^{\hat H} \hat\xi$ is positive. Indeed, using the definition of $\cA^{\hat H}{\hat\xi}$ (cf.~\eqref{eqn:AGF}), \eqref{eqn:xihat ODE}, \eqref{eqn:prop:ELMM:relations:H}, the fact that ${a(v,\hat y(v),p) = 1- \frac{\phi'(v)}{\kappa^G(v)}\bar\pi(v,v)}$ by the definitions of $\bar\pi$ and $a$ in \eqref{eqn:pibar} and \eqref{eqn:a}, and \eqref{eqn:auxiliary functions >0},
\begin{align}
\label{eqn:AHxihat}
\cA^{\hat H} \hat\xi(v)
&= \hat\xi(v) - \frac{\hat\xi'(v)}{\kappa^{\hat H}(v)}
= \hat\xi(v)\left (1-\bar\pi(v,v)\frac{\phi'(v)}{\kappa^G(v)} \right)
= \hat\xi(v)a(v,\hat y(v), p) > 0, \quad  v \in (0, T).%
\end{align}
\end{proof}

\begin{lemma}
\label{lem:OC1}
The triplet $(\hat \pi, \hat Q, \hat z)$ defined in the proof of Theorem~\ref{thm:main result} satisfies $U'(X^{\hat\pi}_T) = \hat z \frac{\diff \hat Q}{\diff P}$.
\end{lemma}

\begin{proof}
By Lemma~\ref{lem:wealth process} and the fact that $U'(x) = x^{-p}$,
\begin{align}
\label{eqn:lem:OC1:pf:first step}
U'(X^{\hat\pi}_T)
&= x^{-p} \hat X_T^{-p}
= x^{-p} \cE_T\left ( \sigma \int_0^\cdot \hat\pi_t \dd W^{\hat Q}_t \right)^{-p} \left (\mart[T]{\hat H}{ \hat\xi}\right )^{-p}.
\end{align}
First, a standard calculation gives
\begin{align}
\label{eqn:lem:OC1:pf:first factor}
\cE_T\left ( \sigma \int_0^\cdot \hat\pi_t \dd W^{\hat Q}_t \right)^{-p} &= \cE_T \left ( -p\sigma \int_0^\cdot \hat\pi_t \dd W_t \right ) \exp \left ((1-p)\frac{p \sigma^2}{2}  \int_0^T \hat\pi_t^2 \dd t \right ).
\end{align}
To compute the second factor, we claim that for $v \in [0,T)$,
\begin{align}
\label{eqn:lem:OC1:pf:xihat zetahat}
\hat\xi (v)m(v,\hat y(v),p)
&= x^{-1} \hat z^{-\frac{1}{p}}\exp\left((1-p) \frac{\sigma^2}{2} \int_0^T \bar\pi(u,v)^2 \dd u\right ) \hat\zeta(v)^{-\frac{1}{p}},
\end{align}
where $\bar\pi$ is defined in \eqref{eqn:pibar} and $\hat\zeta$ is given in \eqref{eqn:thm:ELMM:zeta} (with $y$ replaced by $\hat y$). Moreover, in the case $\Delta G(T) >0$, we claim that
\begin{align}
\label{eqn:lem:OC1:pf:xihat zetahat T}
\hat\xi (T-) 
&= x^{-1} \hat z^{-\frac{1}{p}}\exp\left ((1-p) \frac{\sigma^2}{2} \int_0^T \bar\pi(u,u)^2 \dd u\right ) \hat\zeta(T-)^{-\frac{1}{p}}.
\end{align}
Then, by \eqref{eqn:AHxihat}, the definition of $m$ in \eqref{eqn:m}, \eqref{eqn:lem:OC1:pf:xihat zetahat}, and \eqref{eqn:prop:ELMM:relations:AGzeta}, on $\{\gamma < T\}$,
\begin{align}
\left (\mart[T]{\hat H}{\hat\xi} \right )^{-p}
&= \left (\cA^{\hat H} \hat\xi(\gamma)\right )^{-p} = \hat\xi(\gamma)^{-p} a(\gamma,\hat y(\gamma),p)^{-p} \notag\\
&= \left (\hat\xi(\gamma) m(\gamma,\hat y(\gamma),p)\right )^{-p}(1+\hat y(\gamma))\notag\\
&= x^p \hat z \exp\left ((1-p) \frac{\sigma^2}{2} \int_0^T \hat\pi_t^2 \dd t \right)^{-p} \hat\zeta(\gamma) (1+\hat y(\gamma)) \notag\\
\label{eqn:lem:OC1:pf:second factor}
&= x^p \hat z\exp\left(-(1-p) \frac{p\sigma^2}{2} \int_0^T \hat\pi_t^2 \dd t \right)\cA^G \hat\zeta(\gamma).
\end{align}
If $\Delta G(T) > 0$, then $\cA^{\hat H} \hat\xi(\gamma) = \hat\xi(T-)$ and $\cA^G \hat\zeta(\gamma) = \hat\zeta(T-)$ on $\lbrace \gamma = T \rbrace$. This together with \eqref{eqn:lem:OC1:pf:xihat zetahat T} shows that \eqref{eqn:lem:OC1:pf:second factor} holds on $\{\gamma = T\}$, too.

Finally, plugging \eqref{eqn:lem:OC1:pf:first factor} and \eqref{eqn:lem:OC1:pf:second factor} into \eqref{eqn:lem:OC1:pf:first step} yields by the definitions of $\hat\pi$ in \eqref{eqn:thm:main result:trading strategy} and $\frac{\diff \hat Q}{\diff P}$ in Theorem~\ref{thm:ELMM} (cf.~\eqref{eqn:thm:ELMM:Z}) that
\begin{align*} 
U'(X^{\hat\pi}_T)
&= \hat z \cE_T \left ( -p\sigma \int_0^\cdot \hat\pi_t \dd W_t \right ) \cA^G \hat\zeta(\gamma) = \hat z \cE_T \left ( -p\sigma \int_0^\cdot \hat\pi_t \dd W_t \right ) \mart[T]{G}{\hat\zeta} \\
&= \hat z\cE_T\left ( -\int_0^\cdot \frac{1}{\sigma} \left(\mu- \phi'(t)\hat y(t)\1_{\lbrace t \leq \gamma, t < T\rbrace}\right) \dd W_t \right ) \mart[T]{G}{\hat\zeta}
= \hat z \frac{\diff \hat Q}{\diff P}.
\end{align*}

It remains to show \eqref{eqn:lem:OC1:pf:xihat zetahat} and \eqref{eqn:lem:OC1:pf:xihat zetahat T}. First, an easy but tedious calculation using the definitions of $\hat\pi$ and $n$ in \eqref{eqn:pibar} and \eqref{eqn:n} shows that for $u \in [0,T)$,
\begin{align}
\label{eqn:lem:OC1:pf:computation}
\phi'(u) \bar\pi(u,u) (1+\hat y(u)) + n(u, \hat y(u),p)
&= \frac{1-p}{2 p^2 \sigma^2} \phi'(u)\hat y(u)\left(\phi'(u)\hat y(u)-2\mu \right) +\frac{1}{p}\kappa^G(u)\hat y(u).
\end{align}
Next, fix $v \in [0,T)$. Using first that $\hat y$ is a solution to the integral equation \eqref{eqn:thm:main result:integral equation} and the definitions of $\hat\xi$ in \eqref{eqn:xihat} and $\bar\pi$ in \eqref{eqn:pibar}, then \eqref{eqn:lem:OC1:pf:computation}, and finally again the definition of $\bar\pi$ and the definition of $\hat \zeta$ in \eqref{eqn:thm:ELMM:zeta} (with $y$ replaced by $\hat y$),
\begin{align*}
\hat\xi (v) \frac{m(v,\hat y(v),p)}{m(0,\hat y(0),p)}
&= \exp\left(\int_0^v \left(\phi'(u)\bar\pi(u,u) (1 + \hat y(u)) + n(u,\hat y(u),p)\right) \dd u  \right ) \\
&= \exp\left(\int_0^v \left(\frac{1-p}{2 p^2 \sigma^2} \phi'(u)\hat y(u)\left(\phi'(u)\hat y(u)-2\mu \right) +\frac{1}{p}\kappa^G(u)\hat y(u) \right) \dd u \right)\\
&= \exp\left ( -\frac{1-p}{2 p^2 \sigma^2} \mu^2 T + (1-p) \frac{\sigma^2}{2} \int_0^T \bar\pi(u,v)^2 \dd u\right ) \hat\zeta(v)^{-\frac{1}{p}},
\end{align*}
and using the definition of $\hat z$ in \eqref{eqn:zhat} gives \eqref{eqn:lem:OC1:pf:xihat zetahat}.

Finally, assume that $\Delta G(T) > 0$. Then also $\Delta \hat H(T) > 0$ since $\hat Q \approx P$. Moreover, by the proof of Theorem~\ref{thm:ELMM}, $\mart{G}{\hat \zeta}$ is positive, and by Lemma~\ref{lem:wealth process}, $\mart{\hat H}{\hat \xi}$ is positive. This together with Proposition~\ref{prop:local martingale property}~(a) and (b)~(i) implies that the limits $\hat\zeta(T-)$ and $\hat\xi(T-)$ exist in $\RR$.  Moreover, $\lim_{v \uparrow \uparrow T} m(v,\hat y(v),p) = 1$ by \eqref{eqn:thm:main result:integral equation} (recall that Theorem~\ref{thm:integral equation} shows that $\hat y$ is a solution to the integral equation) and thus \eqref{eqn:lem:OC1:pf:xihat zetahat T} follows from taking the limit $v \uparrow\uparrow T$ in \eqref{eqn:lem:OC1:pf:xihat zetahat}; the exchange of limit and integration on the right-hand side is justified by dominated convergence using the estimate $\vert \bar\pi(u,v)\vert \leq \frac{1}{p \sigma^2} ( \mu + \vert \phi'(u) \hat y(u) \vert)$ and \eqref{eqn:thm:integral equation:integrability properties}; also note that for $u\in[0,T)$, $\lim_{v \uparrow\uparrow T} \bar\pi(u,v) = \frac{1}{p \sigma^2}(\mu - \phi'(u)\hat y(u)) = \bar\pi(u,u)$.
\end{proof}

\begin{lemma}
\label{lem:OC2}
The triplet $(\hat \pi, \hat Q, \hat z)$ defined in the proof of Theorem~\ref{thm:main result} satisfies $\EX[\hat Q]{X^{\hat \pi}_T}=x$.
\end{lemma}

\begin{proof}
It suffices to show that $X^{\hat\pi}$ is a $\hat Q$-martingale. Lemma~\ref{lem:wealth process} shows that $X^{\hat\pi}$ is of the form \eqref{eqn:cor:ELMM:structure}. Therefore, by Corollary~\ref{cor:ELMM:structure}, it suffices to prove that $\int_0^T \vert\cA^{\hat H} \hat \xi(u)\vert \hat H'(u) \dd u < \infty$ and that $\mart{\hat H}{\hat\xi}$ is a $\hat Q$-martingale. The first assertion follows directly from Proposition~\ref{prop:local martingale property}~(a) noting that $\mart{\hat H}{\hat\xi}$ is positive by Lemma~\ref{lem:wealth process}. For the second assertion, we note that by Lemma~\ref{lem:filtration}~(b)~(ii), it is enough to show that $\mart{\hat H}{\hat\xi}$ is a $(Q^\gamma,\FF^\gamma)$-martingale. To this end, by Proposition~\ref{prop:local martingale property}~(a) and (b), we may assume that $\Delta G(T) = 0$ (using that $\hat Q \approx P$) and have to check that $\lim_{t \uparrow \uparrow T} \hat \xi (t)(1- \hat H(t)) = 0$. We distinguish two cases.

First, let $p \geq 1$ and fix $t \in [0,T)$. Then as $1 - \hat H(t) \leq 1$,
\begin{align*}
0
&\leq \hat \xi (t)(1- \hat H(t))
\leq \hat \xi (t)(1- \hat H(t))^{1/p}
\end{align*}
and it suffices to show that the right-hand side converges to $0$ as $t\uparrow\uparrow T$.
Using first the definitions of $\hat\xi$ and $\hat H$ in \eqref{eqn:xihat} and \eqref{eqn:thm:ELMM:H}, then the definition of $a(\cdot, \cdot, 1)$ in \eqref{eqn:a}, and finally the definition of $b(\cdot,\cdot,1)$ in \eqref{eqn:b}, 
\begin{align}
\hat \xi (t)(1- \hat H(t))^{1/p}
&= \exp \left (\int_0^t \Big( \phi'(u)\frac{1}{p \sigma^2}(\mu - \phi'(u) \hat y(u))(1+\hat y(u)) - \frac{\kappa^G (u)}{p}(1 + \hat y(u)) \Big) \dd u \right) \notag\\
&= \exp \left (-\int_0^t \frac{\kappa^G (u)}{p} (1 + \hat y(u))\Big( 1 - \frac{\phi'(u)}{\kappa^G(u)} \frac{1}{\sigma^2}(\mu - \phi'(u) \hat y(u)) \Big) \dd u \right)\notag\\
&= \exp \left (-\int_0^t \frac{\kappa^G (u)}{p} (1 + \hat y(u))a(u,\hat y(u),1) \dd u \right) \notag\\
\label{eqn:lem:OC2:pf:high p:calculation}
&= \exp \left (-\int_0^t \frac{\kappa^G (u)}{p}b(u,\hat y(u),1) \dd u \right).
\end{align}
By the representation of $n$ in \eqref{eqn:n identity}, we have for $u \in [0,T)$,
\begin{align*}
n(u,\hat y(u),p)
&=  -\frac{1 - p}{2 p^2 \sigma^2} \mu^2 + \frac{1 - p}{2 p^2 \sigma^2} \left(\phi'(u) \hat y(u) - \mu\right)^2 + \frac{1}{p} \kappa^G(u)(b(u,\hat y(u),1) -1).
\end{align*}
As the left-hand side as well as the first two summands on the right-hand side are integrable on $(0,T)$ by \eqref{eqn:thm:integral equation:integrability properties}, we infer that $\int_0^T \kappa^G(u)\vert b(u,\hat y(u),1)- 1 \vert\dd u < \infty$. But $\Delta G(T) = 0$ implies that
\begin{align}
\label{eqn:kappa:nonintegrable}
\int_0^T \kappa^G(u) \dd u
&= -\log (\Delta G(T))
= \infty,
\end{align}
and so the right-hand side of \eqref{eqn:lem:OC2:pf:high p:calculation} converges to $0$ as $t \uparrow\uparrow T$.

Second, let $p < 1$ and fix $t\in[0,T)$. Using first the definitions of $\hat\xi$ and $\hat H$ in \eqref{eqn:xihat} and \eqref{eqn:thm:ELMM:H}, and then the definition of $a$ in \eqref{eqn:a},
\begin{align*}
\hat \xi (t)(1- \hat H(t))
&= \exp \left (\int_0^t \Big( \phi'(u)\frac{1}{p \sigma^2}(\mu - \phi'(u) \hat y(u))(1+\hat y(u)) - \kappa^G (u)(1 + \hat y(u)) \Big) \dd u \right)\\
&= \exp \left (-\int_0^t \kappa^G (u) (1 + \hat y(u))\Big( 1 - \frac{\phi'(u)}{\kappa^G(u)} \frac{1}{p \sigma^2}(\mu - \phi'(u) \hat y(u)) \Big) \dd u \right) \\
&= \exp\left(- \int_0^t \kappa^G(u)(1 + \hat y(u))a(u, \hat y(u), p) \dd u \right).
\end{align*}
Using the estimate
\begin{align*}
(1 + \hat y(u))a(u, \hat y(u), p)
&\geq p \left(1 + \frac{1}{p} \hat y(u)\right) a(u, \hat y(u), p)
= p b(u, \hat y(u), p), \quad u \in [0, T),
\end{align*}
we obtain
\begin{align}
0
&\leq \hat \xi (t)(1- \hat H(t))
\leq \exp\left(- p\int_0^t \kappa^G(u)b(u, \hat y(u), p) \dd u \right).
\label{eqn:lem:OC2:pf:low p:estimate}
\end{align}
By the definition of $n$ in \eqref{eqn:n}, we have for $u \in [0,T)$,
\begin{align*}
n(u,\hat y(u),p)
= -\frac{1-p}{2 p^2 \sigma^2} \left(\phi'(u) \hat y(u) \right)^2 + \kappa^G(u)(b(u,\hat y(u),p) -1).
\end{align*}
As the left-hand side as well as the first summand on the right-hand side are integrable on $(0,T)$ by \eqref{eqn:thm:integral equation:integrability properties}, we infer that $\int_0^T \kappa^G(u)\vert b(u,\hat y(u),p)- 1 \vert\dd u < \infty$. Combining this with \eqref{eqn:kappa:nonintegrable} shows that the right-hand side of \eqref{eqn:lem:OC2:pf:low p:estimate} converges to $0$ as $t\uparrow\uparrow T$.
\end{proof}

\small
\providecommand{\bysame}{\leavevmode\hbox to3em{\hrulefill}\thinspace}
\providecommand{\MR}{\relax\ifhmode\unskip\space\fi MR }
\providecommand{\MRhref}[2]{%
  \href{http://www.ams.org/mathscinet-getitem?mr=#1}{#2}
}
\providecommand{\href}[2]{#2}

\end{document}